\documentclass[12pt]{amsart}
\usepackage{}

\usepackage{amsmath}
\usepackage{amsfonts}
\usepackage{amssymb}
\usepackage[all]{xy}           %xypic macro for latex2.09

\usepackage{bbding}
\usepackage{txfonts}
\usepackage{amscd}

\usepackage[shortlabels]{enumitem}
\usepackage{ifpdf}
\ifpdf
  \usepackage[colorlinks,final,backref=page,hyperindex]{hyperref}
\else
  \usepackage[colorlinks,final,backref=page,hyperindex,hypertex]{hyperref}
\fi
\usepackage{tikz}
\usepackage[active]{srcltx}
\usepackage{cancel}
\makeatletter

\newtheorem{thm}{Theorem}[section]
\newtheorem{lem}[thm]{Lemma}
\newtheorem{cor}[thm]{Corollary}
\newtheorem{pro}[thm]{Proposition}
\newtheorem{ex}[thm]{Example}
\newtheorem{rmk}[thm]{Remark}
\newtheorem{defi}[thm]{Definition}

\newcommand {\emptycomment}[1]{}

\newcommand{\lon }{\,\rightarrow\,}
\newcommand{\be }{\begin{equation}}
\newcommand{\ee }{\end{equation}}

\newcommand{\g}{\mathfrak g}

\newcommand{\huaB}{\mathcal{B}}%{{\mathcal{E}}}%{\mathcal{B}}

\newcommand{\huaH}{\mathcal{H}}

\newcommand{\huaZ}{\mathcal{Z}}

\newcommand{\frkR}{\mathfrak R}

\newcommand{\Courant}[1]{\left\llbracket  #1\right\rrbracket }

\newcommand{\Id}{{\rm{Id}}}

\newcommand{\br}[1]{   [ \cdot,    \cdot  ]   }

\newcommand{\Hom}{\mathrm{Hom}}

\newcommand{\Der}{\mathrm{Der}}

\newcommand{\gl}{\mathfrak {gl}}

\newcommand{\ad}{\mathrm{ad}}

\newcommand{\sgn}{\mathrm{sgn}}

\newcommand{\Ni}{\mathsf{nLie}}

\begin{document}

\title{Reynolds $n$-Lie algebras and NS-$n$-Lie algebras}

\author{Shuai Hou}
\address{Department of Mathematics, Jilin University, Changchun 130012, Jilin, China}
\email{houshuai19@mails.jlu.edu.cn}

\author{Yunhe Sheng}
\address{Department of Mathematics, Jilin University, Changchun 130012, Jilin, China}
\email{shengyh@jlu.edu.cn}

%\date{\today}

\begin{abstract}
In this paper, first we introduce the notion of a Reynolds operator on an $n$-Lie algebra and illustrate the relationship between Reynolds operators and derivations on an $n$-Lie algebra. We give the cohomology theory of Reynolds operators on an $n$-Lie algebra and study infinitesimal deformations of Reynolds operators using the second  cohomology group. Then we introduce the notion of    NS-$n$-Lie algebras, which are generalizations  of both $n$-Lie algebras and $n$-pre-Lie algebras. We show that an NS-$n$-Lie algebra gives rise to  an $n$-Lie algebra together with a representation on itself.   Reynolds operators and   Nijenhuis operators on an $n$-Lie algebra naturally induce  NS-$n$-Lie algebra structures. Finally, we construct Reynolds $(n+1)$-Lie algebras and Reynolds $3$-Lie algebras from Reynolds $n$-Lie algebras and Reynolds commutative associative algebras respectively.
\end{abstract}



\maketitle

\tableofcontents

\allowdisplaybreaks

%\end{document}

\section{Introduction}

%\subsection{Reynolds operators }

Reynolds operators occurred for the first time in O. Reynolds' famous study  of turbulence theory into fluid dynamics (\cite{Re}). In the turbulent flows models of fluid dynamics, especially in the
Reynolds-averaged Navier-Stokes equations, Reynolds operators   often take the average over the fluid flow under the group of time translations.
Afterwards, the Reynolds operator was named  in \cite{Ka1}, where
the operator was considered  as a mathematical subject in general.
Let $A$ be an algebra over a filed. A Reynolds operator is linear map $R:A\rightarrow A$ satisfying the following identity
\begin{eqnarray*}
R(xy)=RxRy+R\Big((x-Rx)(y-Ry)\Big),\quad \forall x,y\in A.
\end{eqnarray*}
The identity is called the {\bf Reynolds identity}. In \cite{Dubreil-Jacotin}, the authors expanded the Reynolds identity and got the equivalent form
\begin{eqnarray*}
R(x)R(y)=R(xR(y))+R(R(x)y)-R(R(x)R(y)),\quad \forall x,y\in A.
\end{eqnarray*}
See \cite{FM,Mi1,Mi2,Ro2} for more studies of Reynolds operators. In particular, the
  free Reynolds algebras were given in \cite{gao-guo}. Recently,  A. Das introduced the notion of a Reynolds operators on a Lie algebra in \cite{Das-1} in the study of twisted Rota-Baxter operators.

The notion of an $n$-Lie algebra was introduced  by Filippov  in \cite{Filippov}. The $n$-Lie algebra is the algebraic
structure corresponding to Nambu mechanics \cite{N}. Recently,  $n$-Lie algebras have been widely studied on account of its appearing in many fields of mathematics and physics \cite{Bagger-J-1,Bagger-J}. For example, when $n=3$, $3$-Lie algebras are related to the study of  supersymmetry and gauge symmetry transformations of the word-volume theory of multiple M2-branes. See \cite{review} for more details about $n$-Lie algebras.
The notion of a Nijenhuis operator on an $n$-Lie algebra was introduced
in  \cite{Liu-Jie-Feng} to study deformations of $n$-Lie algebras. It is well-known that  Nijenhuis operators play a
significant role in deformation theories on account of their relationship with trivial infinitesimal deformations.
Deformations of $n$-Lie algebras have been studied extensively in \cite{Arfa,Makhlouf,Takhtajan1}.
In \cite{LP}, Leroux introduced the notion of an NS-algebra which consists of two binary operations and show that a Nijenhuis operator can define an NS-algebra.
In \cite{LG}, the authors studied the relationship
between the category of Nijenhuis algebras and the category of NS-algebras. Recently, the notion  of an NS-Lie algebra was also  introduced in \cite{Das-1}. A Reynolds operators on a Lie algebra, more generally a twisted Rota-Baxter operator on a Lie algebra induces an NS-Lie algebra naturally.

In this paper, first we introduce the notion of a Reynolds operator on an $n$-Lie algebra. Note that there are close relationships between Reynolds operators and derivations on $n$-Lie algebras.  We study the cohomology theory of Reynolds operator   and  use the second cohomology group to study
infinitesimal deformations of Reynolds operators on $n$-Lie algebras.  Then
we introduce a new algebraic structure which is called  an NS-$n$-Lie algebra. An NS-$n$-Lie algebra naturally gives rise to an $n$-Lie algebra and a representation on itself. We show that a Reynolds operator and a Nijenhuis operator on an $n$-Lie algebra
naturally induce an NS-$n$-Lie algebra respectively. Finally, according to constructions of $n$-Lie algebras, we give various constructions of Reynolds $n$-Lie algebras.

The paper is organized as follows. In Section \ref{sec:L},  we introduce the notion of a Reynolds operator on an $n$-Lie algebra and illustrate the relationship between Reynolds operators and derivations. Furthermore, we establish the cohomology theory of a Reynolds operator on an $n$-Lie algebra. Applications are given to study infinitesimal deformations of Reynolds operators. In Section \ref{sec:GM}, we introduce the notion of
an NS-$n$-Lie algebra and show that a Reynolds operator and a Nijenhuis operator on an $n$-Lie algebra naturally
induce an NS-$n$-Lie algebra respectively. In Section \ref{sec:Con},  we construct Reynolds $(n+1)$-Lie algebras and Reynolds $3$-Lie algebras from Reynolds $n$-Lie algebras and Reynolds  commutative associative algebras respectively.

\vspace{2mm}

%In this paper, we work over an algebraically closed field $\K$ of characteristic 0 and all the vector spaces are over $\K$ and finite-dimensional.

\vspace{2mm}
\noindent
{\bf Acknowledgements. } This research is  supported by National Science Foundation of China
(11922110).

\section{Cohomologies and deformations of Reynolds operators on $n$-Lie algebras}\label{sec:L}
\subsection{Reynolds operators on $n$-Lie algebras}

In this subsection, we introduce the notion of a Reynolds operator on an $n$-Lie algebra and provide the replicating property of Reynolds $n$-Lie algebras.
We also find that there is a close relationship between Reynolds operators and derivations on an $n$-Lie algebra.
\begin{defi}{\rm\cite{Filippov}}
An {\bf$n$-Lie algebra} is a vector space $\g$
together with a skew-symmetric linear multiplication $[\cdot,\cdots,\cdot]_{\g}: \wedge^n \g\rightarrow \g$ such that for all $x_i,y_i\in \g, 1\leq i\leq n$, the following Filippov Identity is satisfied:
\begin{equation}\label{FI-Identity}
[x_1, \cdots, x_{n-1}, [y_1, \cdots, y_n]_{\g}]_{\g}=\sum_{i=1}^n[y_1, \cdots, y_{i-1}, [x_1, \cdots, x_{n-1},y_{i}]_{\g},y_{i+1},\cdots, y_n]_{\g}.
\end{equation}
\end{defi}
A {\bf derivation} on $\g$ is a linear map $D:\g\rightarrow\g$ such that
\begin{equation}\label{eq:der}
 D([x_1, \cdots, x_n]_{\g})=\sum_{i=1}^n[x_1, \cdots, D(x_i), \cdots, x_n]_{\g},\quad\forall x_1, \cdots, x_n\in \g.
\end{equation}

For $x_1,x_2,\cdots,x_{n-1}\in \g,$ define $\ad:\wedge ^{n-1}\g\rightarrow \gl(\g)$ by
$$\ad_{x_1,x_2,\cdots,x_{n-1}}y:=[x_1,x_2,\cdots,x_{n-1},y]_{\g},\quad \forall y\in \g.$$
Then $\ad_{x_1,x_2,\cdots,x_{n-1}}$ is a derivation, i.e.
$$\ad_{x_1,x_2,\cdots,x_{n-1}}[y_1, \cdots, y_{n}]_{\g}=\sum^{n}_{i=1}[y_1, \cdots, y_{i-1},\ad_{x_1,x_2,\cdots,x_{n-1}}y_{i},y_{i+1},\cdots,y_{n}]_{\g}.$$

\begin{defi}{\rm\cite{KA}}\label{defi-representation}
Let $V$ be a vector space. A {\bf representation} of an $n$-Lie algebra $(\g,[\cdot,\cdots,\cdot]_{\g})$ on $V$ is a multilinear map $\rho:\wedge^{n-1}\g\rightarrow \gl(V)$, such that for all $x_1,\cdots,x_{n-1},y_1,\cdots,y_n\in \g$, the following equalities hold:
\begin{eqnarray}
 \label{n-representation-1}[\rho(\mathfrak{X}),\rho(\mathfrak{Y})]_{\g}&=&\rho(\mathfrak{X}\circ\mathfrak{Y}),\\
  \label{n-representation-2}\rho(x_1,\cdots,x_{n-2},[y_{1},\cdots,y_{n}]_{\g})&=&\sum_{i=1}^n(-1)^{n-i}\rho (y_{1},\cdots,\widehat{y_{i}},\cdots y_{n})\rho (x_{1},\cdots,x_{n-2},y_i),
  \end{eqnarray}
  where $\mathfrak{X}\circ \mathfrak{Y}=\sum\limits_{i=1}^{n-1}y_1\wedge\cdots\wedge y_{i-1}\wedge[x_1,\cdots,x_{n-1},y_i]_{\g}\wedge y_{i+1}\wedge\cdots\wedge y_{n-1}$ and $\mathfrak{X}=x_1\wedge \cdots \wedge x_{n-1},\mathfrak{Y}=y_1\wedge \cdots \wedge y_{n-1}$.
\end{defi}

Given a representation $(V;\rho)$, there is a semi-direct product $n$-Lie algebra structure on $\g\oplus V$ given by
\begin{eqnarray*}
[x_1+v_1,\cdots,x_n+v_n]_{\g\oplus V}=[x_1,\cdots,x_n]_{\g}+\sum\limits_{i=1}^{n}(-1)^{n-i}\rho(x_1,\cdots,\widehat{x_{i}},\cdots,x_{n})(v_i),\quad x_i\in \g,v_i\in V,
\end{eqnarray*}
where $\widehat{x_{i}}$ means that $x_i$ is omitted.

We denote this semi-direct product $n$-Lie algebra by $\g\ltimes_{\rho} V.$ In particular, when $n=2,$ i.e. for a Lie algebra, we obtain the usual notion of a semi-direct product Lie algebra.

 Let $(V;\rho)$ be a representation of an $n$-Lie algebra $(\g,[\cdot,\cdots,\cdot]_{\g})$. Denote the space of $m$-cochains by
$$C_{\Ni}^{m}(\g;V)=\Hom (\otimes^{m-1} (\wedge^{n-1}\g)\wedge\g,V),\quad(m\geq1).$$

The coboundary operator ${\partial_{\rho}}:C_{\Ni}^{m}(\g;V)\rightarrow C_{\Ni}^{m+1}(\g;V)$ is defined by
\begin{eqnarray*}&&
({\partial_{\rho}}f)(\mathfrak{X}_1,\cdots,\mathfrak{X}_m,x_{m+1})\\
&=&\sum_{1\leq j<k\leq m}(-1)^{j} f(\mathfrak{X}_1,\cdots,\widehat{\mathfrak{X}_{j}},\cdots,\mathfrak{X}_{k-1},
\mathfrak{X}_j\circ\mathfrak{X}_k,\mathfrak{X}_{k+1},\cdots,\mathfrak{X}_{m},x_{m+1})\\
&&+\sum_{j=1}^{m}(-1)^{j}f(\mathfrak{X}_1,\cdots,\widehat{\mathfrak{X}_{j}},\cdots,\mathfrak{X}_{m},
[\mathfrak{X}_j,x_{m+1}]_{\g})\\
&&+\sum_{j=1}^{m}(-1)^{j+1}\rho(\mathfrak{X}_j)f(\mathfrak{X}_1,\cdots,\widehat{\mathfrak{X}_{j}},
\cdots,\mathfrak{X}_{m},x_{m+1})\\&&
+\sum_{i=1}^{n-1}(-1)^{n+m-i+1}\rho(x^1_m,\cdots,\widehat{x^i_m},\cdots,x^{n-1}_m, x_{m+1})f(\mathfrak{X}_1,\cdots,\mathfrak{X}_{m-1},x^{i}_m),
\end{eqnarray*}
for any $\mathfrak{X}_{i}=x^1_{i}\wedge\cdots\wedge x^{n-1}_{i}\in \wedge^{n-1}\g, i=1,2,\cdots,m,x_{m+1}\in \g.$
It was proved in {\rm\cite{Casas,Takhtajan1}} that $\partial_{\rho}\circ\partial_{\rho}=0$. Thus, $(\oplus^{+\infty}_{m=1} C_{\Ni}^{m}(\g;V),\partial_{\rho}) $ is a cochain complex.
\begin{defi}
The cohomology of the $n$-Lie algebra $\g$ with coefficients in $V$ is the cohomology of the cochain complex $(\oplus^{+\infty}_{m=1} C_{\Ni}^{m}(\g;V),\partial_{\rho}) $. The corresponding $m$-th cohomology group is denoted by $\huaH_{\Ni}^{m}(\g;V),$ for $m\geq 1.$
\end{defi}

\begin{defi}
Let $(\g,[\cdot,\cdots,\cdot]_{\g})$ be an $n$-Lie algebra. A linear map $R:\g\rightarrow\g$ is called a {\bf Reynolds operator} if
\begin{equation}
\label{n-Reynolds} [Rx_1,\cdots,Rx_n]_{\g}=\sum^n_{i=1}(-1)^{n-i} R[Rx_1,\cdots,\widehat{Rx_i},\cdots,Rx_n,x_i]_{\g}-R[Rx_1,\cdots,Rx_n]_{\g},
\end{equation}
where $x_1,x_2,\cdots,x_n\in \g.$
Moreover, an $n$-Lie algebra $\g$ with a Reynolds operator $R$ is called a {\bf Reynolds $n$-Lie algebra}. We denote it by $(\g,[\cdot,\cdots,\cdot]_{\g},R).$
\end{defi}

\begin{defi}
Let $R$ and $R'$ be Reynolds operators on an $n$-Lie algebra $\g.$ A {\bf homomorphism} of
Reynolds operators from $R$ to $R'$ consists of a pair $(\phi,\varphi)$ of an $n$-Lie algebra homomorphism $\phi:\g\rightarrow \g$ and a linear map $\psi:\g\rightarrow \g$ satisfying
\begin{eqnarray}\label{condition-1}
\phi\circ R=R'\circ \psi,
\end{eqnarray}
i.e. we have the following commutative diagram
\[\xymatrix{
 \mathfrak g \ar[d]_{R} \ar[r]^{\psi}
                & \mathfrak g \ar[d]^{R'}  \\
  \mathfrak g \ar[r]^{\phi}
                & \mathfrak g             }\].

\end{defi}

\begin{thm}\label{Reynolds-n-Lie algebra}
Let $(\g,[\cdot,\cdots,\cdot]_{\g},R)$ be a Reynolds $n$-Lie algebra. Define   $[\cdot,\cdots,\cdot]_{R}:\wedge^n\g\lon \g$   by
\begin{equation}
\label{induce-n-Lie}[x_1,\cdots,x_n]_{R}=\sum^n_{i=1}(-1)^{n-i}[Rx_1,\cdots,\widehat{Rx_i},\cdots,Rx_n,x_i]_{\g}-[Rx_1,\cdots,Rx_n]_{\g},
\end{equation}
for all $x_1,x_2,\cdots,x_n\in \g.$ Then
\begin{itemize}
\item[{\rm (a)}] $[Rx_1,\cdots,Rx_n]_{\g}=R([x_1,\cdots,x_n]_{R})$;
\item[{\rm (b)}] $(\g,[\cdot,\cdots,\cdot]_{R})$ is an $n$-Lie algebra;
\item[{\rm (c)}] $(\g,[\cdot,\cdots,\cdot]_{R},R)$ is a Reynolds $n$-Lie algebra;
\item[{\rm (d)}] The pair $(R,R)$ is a Reynolds $n$-Lie algebra homomorphism from $(\g,[\cdot,\cdots,\cdot]_{R},R)$ to $(\g,[\cdot,\cdots,\cdot]_{\g},R)$.
\end{itemize}
\end{thm}
\begin{proof}
\item[{\rm (a)}.] It follows directly from \eqref{n-Reynolds}.
\item[{\rm (b)}.] It is straightforward to deduce that $[\cdot,\cdots,\cdot]_{R}$ is skew-symmetric.
For $x_1,\cdots,x_{n-1},y_1,\cdots,y_n\in \g,$ by \eqref{FI-Identity} and \eqref{induce-n-Lie}, we have
\begin{eqnarray*}
&&[x_1,\cdots,x_{n-1},[y_1,\cdots,y_n]_{R}]_{R}-\sum\limits^{n}_{i=1}[y_1,\cdots,y_{i-1},[x_1,\cdots,x_{n-1},y_i]_{R},y_{i+1},\cdots,y_n]_{R}\\
&=&[Rx_1,\cdots,Rx_{n-1},[y_1,\cdots,y_n]_{R}]_{\g}+\sum\limits^{n-1}_{i=1}(-1)^{n-i}[Rx_1,\cdots,\widehat{Rx_i},\cdots,Rx_{n-1},[Ry_1,\cdots,Ry_n]_{\g},x_i]_{\g}\\
&&-[Rx_1,\cdots,Rx_{n-1},[Ry_1,\cdots,Ry_n]_{\g}]_{\g}-\sum\limits^{n}_{i=1}[Ry_1,\cdots,Ry_{i-1},[x_1,\cdots,x_{n-1},y_{i}]_{R},Ry_{i+1},\cdots,Ry_{n}]_{\g}
\\&&+\sum\limits^{n}_{i=1}[Ry_1,\cdots,Ry_{i-1},[Rx_1,\cdots,Rx_{n-1},Ry_i]_{\g},Ry_{i+1},\cdots,Ry_{n}]_{\g}\\
&&-\sum\limits^{n}_{i=1}\sum\limits_{j=1,j\neq i}^{n}(-1)^{n+i-j-1}[[Rx_1,\cdots,Rx_{n-1},Ry_i]_{\g},Ry_1,\cdots,,\widehat{Ry_{j}},\cdots,Ry_{n},y_{j}]_{\g}\\
&=&\sum^n_{i=1}(-1)^{n-i}[Rx_1,\cdots,Rx_{n-1},[Ry_1,\cdots,\widehat{Ry_i},\cdots,Ry_n,y_i]_{\g}]_{\g}-2[Rx_1,\cdots,Rx_{n-1},[Ry_1,\cdots,Ry_n]_{\g}]_{\g}\\
&&+\sum\limits^{n-1}_{i=1}(-1)^{n-i}[Rx_1,\cdots,\widehat{Rx_i},\cdots,Rx_{n-1},[Ry_1,\cdots,Ry_n]_{\g},x_i]_{\g}\\
&&-\sum\limits^{n}_{i=1}\sum\limits^{n-1}_{j=1}(-1)^{n-j}[Ry_1,\cdots,Ry_{i-1},[Rx_1,\cdots,\widehat{Rx_j},\cdots,Rx_{n-1},Ry_{i},x_j]_{\g},Ry_{i+1},\cdots,Ry_{n}]_{\g}\\
&&-\sum\limits^{n}_{i=1}[Ry_1,\cdots,Ry_{i-1},[Rx_1,\cdots,Rx_{n-1},y_{i}]_{\g},Ry_{i+1},\cdots,Ry_{n}]_{\g}\\
&&+\sum\limits^{n}_{i=1}2[Ry_1,\cdots,Ry_{i-1},[Rx_1,\cdots,Rx_{n-1},Ry_{i}]_{\g},Ry_{i+1},\cdots,Ry_{n}]_{\g}\\
&&-\sum\limits^{n}_{i=1}\sum\limits_{j=1,j\neq i}^{n}(-1)^{n+i-j-1}[[Rx_1,\cdots,Rx_{n-1},Ry_i]_{\g},Ry_1,\cdots,\widehat{Ry_{j}},\cdots,Ry_{n},y_{j}]_{\g}\\
&=&0.
\end{eqnarray*}
Thus, $(\g,[\cdot,\cdots,\cdot]_{R})$ is an $n$-Lie algebra.
\item[{\rm (c)}.] For $x_1,x_2,\cdots,x_n\in \g,$ by \eqref{induce-n-Lie} and Item {\rm (a)}, we have
\begin{eqnarray*}
[Rx_1,Rx_2,\cdots,Rx_n]_{R}&=&\sum^n_{i=1}(-1)^{n-i} [R^{2}x_1,\cdots,\widehat{R^{2}x_i},\cdots,R^{2}x_n,Rx_i]_{\g}-[R^{2}x_1,\cdots,R^{2}x_n]_{\g}\\
                           &=&\sum^n_{i=1}(-1)^{n-i}R[Rx_1,\cdots,\widehat{Rx_i},\cdots,Rx_n,x_i]_{R}-R[Rx_1,\cdots,Rx_n]_{R},
\end{eqnarray*}
which implies that $R$ is a Reynolds operator on the $n$-Lie algebra $(\g,[\cdot,\cdots,\cdot]_{R}).$
\item[{\rm (d)}.] By Item {\rm (a)}, $R$ is an $n$-Lie algebra homomorphism. Moreover, $R$ commutes with itself. Therefore,
$(R,R)$ is a Reynolds $n$-Lie algebra homomorphism from the Reynolds $n$-Lie algebra $(\g,[\cdot,\cdots,\cdot]_{R},R)$ to $(\g,[\cdot,\cdots,\cdot]_{\g},R)$
\end{proof}

Next, we study the relationship between Reynolds operators and derivations on an $n$-Lie algebra.
\begin{pro}
Let $R:\g\rightarrow \g$ be a Reynolds operator on an $n$-Lie algebra $(\g,[\cdot,\cdots,\cdot]_{\g})$. If $R$ is invertible, then $(R^{-1}-\frac{1}{n-1}\Id):\g\rightarrow\g$ is a derivation on $\g$, where $\Id$ is the identity operator.
\end{pro}
\begin{proof}
Let $R:\g\rightarrow \g$ be a Reynolds operator on $\g$ such that $R$ is invertible.
By \eqref{n-Reynolds}, we have
\begin{eqnarray*}
R^{-1}[x_1,\cdots,x_n]_{\g}=\sum^n_{i=1} [x_1,\cdots,R^{-1}x_i,\cdots,x_n]_{\g}-[x_1,\cdots,x_n]_{\g},
\end{eqnarray*}
where $x_1,x_2,\cdots,x_n\in \g.$ Then we have

\begin{eqnarray*}
(R^{-1}-\frac{1}{n-1}\Id)[x_1,\cdots,x_n]_{\g}=\sum^n_{i=1} [x_1,\cdots,(R^{-1}-\frac{1}{n-1}\Id)x_i,\cdots,x_n]_{\g}.
\end{eqnarray*}
This shows that $(R^{-1}-\frac{1}{n-1}\Id):\g\rightarrow\g$ is a derivation on $\g$.
\end{proof}

Conversely, we can derive a Reynolds operator on an $n$-Lie algebra from a derivation.
\begin{pro}
Let $D:\g\rightarrow \g$ be a derivation on an $n$-Lie algebra $(\g,[\cdot,\cdots,\cdot]_{\g})$. If $(D+\frac{1}{n-1}\Id):\g\rightarrow\g$ is invertible, then $(D+\frac{1}{n-1}\Id)^{-1}$ is a Reynolds operator.
\end{pro}
\begin{proof}
Let $D:\g\rightarrow \g$ be a derivation on an $n$-Lie algebra $(\g,[\cdot,\cdots,\cdot]_{\g}).$
By \eqref{eq:der}, for $u_1,u_2,\cdots,u_n\in \g$, we have
\begin{eqnarray*}
(D+\frac{1}{n-1}\Id)[u_1,\cdots,u_n]_{\g}=\sum_{i=1}^n[u_1, \cdots, (D+\frac{1}{n-1}\Id)u_i, \cdots, u_n]_{\g}-[u_1,\cdots,u_n]_{\g}.
\end{eqnarray*}
For convenience, we denote $P=D+\frac{1}{n-1}\Id.$
If $P$ is invertible, we put $Pu_i=x_i, 1\leq i\leq n,$
we get
\begin{eqnarray*}
[P^{-1}x_1,\cdots,P^{-1}x_n]_{\g}&=&\sum^n_{i=1}(-1)^{n-i} P^{-1}[P^{-1}x_1,\cdots,\widehat{P^{-1}x_i},\cdots,P^{-1}x_n,x_i]_{\g}\\
&&-P^{-1}[P^{-1}x_1,\cdots,P^{-1}x_n]_{\g},
\end{eqnarray*}
which implies that $P^{-1}$ is a Reynolds operator on an $n$-Lie algebra $(\g,[\cdot,\cdots,\cdot]_{\g}).$ The proof is finished.
\end{proof}

If $(D+\frac{1}{n-1}\Id):\g\rightarrow\g$ is not invertible, by Proposition 2.8 in \cite{gao-guo}, we have the following result.
\begin{pro}\label{invertible-derivation}
Let $\g$ be an $n$-Lie algebra and $D:\g\rightarrow\g$ be a derivation. For each $x\in \g,$ if the infinite sum  $(D+\frac{1}{n-1}\Id)^{-1}(x)=\sum^{\infty}_{m=0}(-1)^{m}(n-1)^{m+1}D^m(x)$ converges to an element in $\g,$ then $R:=\sum^{\infty}_{m=0}(-1)^{m}(n-1)^{m+1}D^m$ is a Reynolds operator on $\g.$
\end{pro}
\begin{proof}
The proof is similar to that of \cite[Proposition 2.8]{gao-guo}, we omit the details.
\end{proof}

By Proposition \ref{invertible-derivation}, if $D$ is a nilpotent (more generally a locally nilpotent) derivation on $\g,$ then for all $x\in \g,$
the series $\sum^{\infty}_{n=0}D^n(x)$ is  a finite sum and convergent, then $R:=\sum^{\infty}_{m=0}(-1)^{m}(n-1)^{m+1}D^m$ is a Reynolds operator on $\g.$

\subsection{Cohomology of Reynolds operators on $n$-Lie algebras}
In this subsection, we construct a representation of the $n$-Lie algebra $(\g,[\cdot,\cdots,\cdot]_{R})$ on the vector space $\g,$ and define the cohomology of Reynolds operators on $n$-Lie algebras.
\begin{lem}\label{representation-R-G}
Let $R:\g\rightarrow\g$ be a Reynolds operator on an $n$-Lie algebra $(\g,[\cdot,\cdots,\cdot]_{\g}).$ Define
$\varrho_{R}:\wedge^{n-1}\g\rightarrow\gl(\g)$ by
 \begin{eqnarray}
 \varrho_{R}(x_1,\cdots,x_{n-1})(x)&=&[Rx_1,\cdots,Rx_{n-1},x]_{\g}+R[Rx_1,\cdots,Rx_{n-1},x]_{\g}\\
\nonumber &&-\sum^{n-1}_{i=1} R[Rx_1,\cdots,Rx_{i-1},x_{i},Rx_{i+1},\cdots,Rx_{n-1},x]_{\g},
  \end{eqnarray}
where $x_1,\cdots,x_{n-1},x\in \g.$ Then $(\g;\varrho_{R})$ is a representation of the $n$-Lie algebra $(\g,[\cdot,\cdots,\cdot]_{R})$.
\end{lem}
\begin{proof}
By \eqref{FI-Identity}, for $x_1,\cdots,x_{n-1},y_1,\cdots,y_{n-1} \in \g,$ we have
{\footnotesize
\begin{eqnarray*}
&&[\varrho_{R}(x_1,\cdots,x_{n-1}),\varrho_{R}(y_1,\cdots,y_{n-1})](x)-\varrho_{R}\Big(\sum\limits_{i=1}^{n-1}(y_1,\cdots,y_{i-1},[x_1,\cdots,x_{n-1},y_i]_{R},y_{i+1},\cdots,y_{n-1})\Big)(x)\\
&=&\varrho_{R}(x_1,\cdots,x_{n-1})\Big([Ry_1,\cdots,Ry_{n-1},x]_{\g}-\sum^{n-1}_{i=1} R[Ry_1,\cdots,y_{i},\cdots,Ry_{n-1},x]_{\g}+R[Ry_1,\cdots,Ry_{n-1},x]_{\g}\Big)\\
&&-\varrho_{R}(y_1,\cdots,y_{n-1})\Big([Rx_1,\cdots,Rx_{n-1},x]_{\g}-\sum^{n-1}_{i=1} R[Rx_1,\cdots,x_{i},\cdots,Rx_{n-1},x]_{\g}+R[Rx_1,\cdots,Rx_{n-1},x]_{\g}\Big)\\
&&-\varrho_{R}\Big(\sum\limits_{i=1}^{n-1}(y_1,\cdots,y_{i-1},[x_1,\cdots,x_{n-1},y_i]_{R},y_{i+1},\cdots,y_{n-1})\Big)(x)\\
&=&[Rx_1,\cdots,Rx_{n-1},[Ry_1,\cdots,Ry_{n-1},x]_{\g}]_{\g}-\sum^{n-1}_{i=1}(-1)^{n-i} R[Rx_1,\cdots,\widehat{Rx_{i}},\cdots,Rx_{n-1},[Ry_1,\cdots,Ry_{n-1},x]_{\g},x_{i}]_{\g}\\
&&+R[Rx_1,\cdots,Rx_{n-1},[Ry_1,\cdots,Ry_{n-1},x]_{\g}]_{\g}-\sum^{n-1}_{i=1}(-1)^{n-i} [Rx_1,\cdots,Rx_{n-1},R[Ry_1,\cdots,\widehat{Ry_{i}},\cdots,Ry_{n-1},x,y_{i}]_{\g}]_{\g}\\
&&+\sum^{n-1}_{j=1}(-1)^{n-j}\sum^{n-1}_{i=1}(-1)^{n-i} R[Rx_1,\cdots,\widehat{Rx_{j}},\cdots,Rx_{n-1}, R[Ry_1,\cdots,\widehat{Ry_{i}},\cdots,Ry_{n-1},x,y_{i}]_{\g},x_{j}]_{\g}\\
&&-\sum^{n-1}_{i=1}(-1)^{n-i}R[Rx_1,\cdots,Rx_{n-1},R[Ry_1,\cdots,\widehat{Ry_{i}},\cdots,Ry_{n-1},x,y_{i}]_{\g}]_{\g}+[Rx_1,\cdots,Rx_{n-1},R[Ry_1,\cdots,Ry_{n-1},x]_{\g}]_{\g}\\
&&-\sum^{n-1}_{i=1}(-1)^{n-i} R[Rx_1,\cdots,\widehat{Rx_i},\cdots,Rx_{n-1},R[Ry_1,\cdots,Ry_{n-1},x]_{\g},x_{i}]_{\g}+R[Rx_1,\cdots,Rx_{n-1},R[Ry_1,\cdots,Ry_{n-1},x]_{\g}]_{\g}\\
&&-[Ry_1,\cdots,Ry_{n-1},[Rx_1,\cdots,Rx_{n-1},x]_{\g}]_{\g}+\sum^{n-1}_{i=1}(-1)^{n-i} R[Ry_1,\cdots,\widehat{Ry_{i}},\cdots,Ry_{n-1},[Rx_1,\cdots,Rx_{n-1},x]_{\g},y_{i}]_{\g}\\
&&-R[Ry_1,\cdots,Ry_{n-1},[Rx_1,\cdots,Rx_{n-1},x]_{\g}]_{\g}+\sum^{n-1}_{i=1}(-1)^{n-i}[Ry_1,\cdots,Ry_{n-1},R[Rx_1,\cdots,\widehat{Rx_{i}},\cdots,Rx_{n-1},x,x_{i}]_{\g}]_{\g}\\
&&-\sum^{n-1}_{j=1}(-1)^{n-j}\sum^{n-1}_{i=1}(-1)^{n-i} R[Ry_1,\cdots,\widehat{Ry_{j}},\cdots,Ry_{n-1},R[Rx_1,\cdots,\widehat{Rx_{i}},\cdots,Rx_{n-1},x,x_{i}]_{\g},y_{j}]_{\g}\\
 &&+\sum^{n-1}_{i=1}(-1)^{n-i}R[Ry_1,\cdots,Ry_{n-1},R[Rx_1,\cdots,\widehat{Rx_{i}},\cdots,Rx_{n-1},x,x_{i}]_{\g}]_{\g}-[Ry_1,\cdots,Ry_{n-1},R[Rx_1,\cdots,Rx_{n-1},x]_{\g}]_{\g}\\
&&+\sum^{n-1}_{i=1}(-1)^{n-i}R[Ry_1,\cdots,\widehat{Ry_{i}},\cdots,Ry_{n-1},R[Rx_1,\cdots,Rx_{n-1},x]_{\g},y_{i}]_{\g}-R[Ry_1,\cdots,Ry_{n-1},R[Rx_1,\cdots,Rx_{n-1},x]_{\g}]_{\g}\\
&&-\sum^{n-1}_{i=1}[Ry_1,\cdots,Ry_{i-1},[Rx_1,\cdots,Rx_{n-1},Ry_i]_{\g},Ry_{i+1},\cdots,Ry_{n-1},x]_{\g}\\
&&+\sum^{n-1}_{i=1}\sum^{n-1}_{j=1,j\neq i}(-1)^{n-j}R[Ry_1,\cdots,\widehat{Ry_{j}},\cdots,[Rx_1,\cdots,Rx_{n-1},Ry_i]_{\g},\cdots,Ry_{n-1},x,y_{j}]_{\g}\\
&&+\sum^{n-1}_{i=1}\sum^{n-1}_{j=1}R[Ry_1,\cdots,Ry_{i-1},[Rx_1,\cdots,x_j,\cdots,Rx_{n-1},Ry_i]_{\g},Ry_{i+1},\cdots,Ry_{n-1},x]_{\g}\\
&&+\sum^{n-1}_{i=1}R[Ry_1,\cdots,Ry_{i-1},[Rx_1,\cdots,Rx_{n-1},y_i]_{\g},Ry_{i+1},\cdots,Ry_{n-1},x]_{\g}\\
&&-\sum^{n-1}_{i=1}R[Ry_1,\cdots,Ry_{i-1},[Rx_1,\cdots,Rx_{n-1},Ry_i]_{\g},Ry_{i+1},\cdots,Ry_{n-1},x]_{\g}\\
&&-\sum^{n-1}_{i=1}R[Ry_1,\cdots,Ry_{i-1},[Rx_1,\cdots,Rx_{n-1},Ry_i]_{\g},Ry_{i+1},\cdots,Ry_{n-1},x]_{\g}\\
&=&0,
\end{eqnarray*}}
and
{\footnotesize
\begin{eqnarray*}
&&\Big(\varrho_{R}(x_1,\cdots,x_{n-2},[y_{1},\cdots,y_{n}]_{R})-\sum_{i=1}^n(-1)^{n-i}\varrho_{R} (y_{1},\cdots,\widehat{y_{i}},\cdots y_{n})\varrho_{R} (x_{1},\cdots,x_{n-2},y_i)\Big)(x)\\
&=&[Rx_1,\cdots,Rx_{n-2},[Ry_{1},\cdots,Ry_{n}]_{\g},x]_{\g}-\sum^{n-2}_{i=1}R[Rx_1,\cdots,x_i,\cdots,Rx_{n-2},[Ry_{1},\cdots,Ry_{n}]_{\g},x]_{\g}\\
&&-R[Rx_1,\cdots,Rx_{n-2},[y_{1},\cdots,y_{n}]_{R},x]+R[Rx_1,\cdots,Rx_{n-2},[Ry_{1},\cdots,Ry_{n}]_{\g},x]_{\g}\\
&&-\sum_{i=1}^n(-1)^{n-i}[Ry_{1},\cdots,\widehat{Ry_i},\cdots,Ry_{n},[Rx_1,\cdots,Rx_{n-2},Ry_i,x]_{\g}]_{\g}\\
&&+\sum_{i=1}^n(-1)^{n-i}\sum_{j=1,j\neq i}^nR[Ry_1,\cdots,\widehat{y_i},\cdots,y_j,\cdots,Ry_n,[Rx_1,\cdots,Rx_{n-2},Ry_i,x]_{\g}]_{\g}\\
&&-\sum_{i=1}^n(-1)^{n-i}R[Ry_{1},\cdots,\widehat{Ry_i},\cdots,Ry_{n},[Rx_1,\cdots,Rx_{n-2},Ry_i,x]_{\g}]_{\g}\\
&&+\sum_{i=1}^n(-1)^{n-i}[Ry_{1},\cdots,\widehat{Ry_i},\cdots,Ry_{n},\sum_{k=1}^{n-2}[Rx_1,\cdots,x_k,\cdots,Rx_{n-2},Ry_i,x]_{\g}]_{\g}\\
&&-\sum_{i=1}^n(-1)^{n-i}\sum_{j=1,j\neq i}^nR[Ry_1,\cdots,\widehat{y_i},\cdots,y_j,\cdots,Ry_n,\sum_{k=1}^{n-2}[Rx_1,\cdots,x_k,\cdots,Rx_{n-2},Ry_i,x]_{\g}]_{\g}\\
&&+\sum_{i=1}^n(-1)^{n-i}R[Ry_{1},\cdots,\widehat{Ry_i},\cdots,Ry_{n},\sum_{k=1}^{n-2}[Rx_1,\cdots,x_k,\cdots,Rx_{n-2},Ry_i,x]_{\g}]_{\g}\\
&&+\sum_{i=1}^n(-1)^{n-i}[Ry_{1},\cdots,\widehat{Ry_i},\cdots,Ry_{n},[Rx_1,\cdots,Rx_{n-2},y_i,x]_{\g}]_{\g}\\
&&-\sum_{i=1}^n(-1)^{n-i}\sum_{j=1,j\neq i}^nR[Ry_1,\cdots,\widehat{y_i},\cdots,y_j,\cdots,Ry_n,[Rx_1,\cdots,Rx_{n-2},y_i,x]_{\g}]_{\g}\\
&&+\sum_{i=1}^n(-1)^{n-i}R[Ry_{1},\cdots,\widehat{Ry_i},\cdots,Ry_{n},[Rx_1,\cdots,Rx_{n-2},y_i,x]_{\g}]_{\g}\\
&&-\sum_{i=1}^n(-1)^{n-i}[Ry_{1},\cdots,\widehat{Ry_i},\cdots,Ry_{n},R[Rx_1,\cdots,Rx_{n-2},Ry_i,x]_{\g}]_{\g}\\
&&+\sum_{i=1}^n(-1)^{n-i}\sum_{j=1,j\neq i}^nR[Ry_1,\cdots,\widehat{y_i},\cdots,y_j,\cdots,Ry_n,R[Rx_1,\cdots,Rx_{n-2},Ry_i,x]_{\g}]_{\g}\\
&&-\sum_{i=1}^n(-1)^{n-i}R[Ry_{1},\cdots,\widehat{Ry_i},\cdots,Ry_{n},R[Rx_1,\cdots,Rx_{n-2},Ry_i,x]_{\g}]_{\g}\\
&=&0.
\end{eqnarray*}}
Therefore, we deduce that  $(\g;\varrho_{R})$ is a representation of the $n$-Lie algebra $(\g,[\cdot,\cdots,\cdot]_{R})$.
\end{proof}
Let ${\rm d_{R}}:C_{\Ni}^{m}(\g;\g)\rightarrow C_{\Ni}^{m+1}(\g;\g),$ $(m\geq 1)$ be the corresponding coboundary operator of the $n$-Lie algebra $(\g,[\cdot,\cdots,\cdot]_{R})$ with coefficients in the representation $(\g;\varrho_{R}).$ More precisely, ${\rm d_{R}}:C_{\Ni}^{m}(\g;\g)\rightarrow C_{\Ni}^{m+1}(\g;\g)$ $(m\geq 1)$ is given by
\begin{eqnarray*}
&&({\rm d_{R}}f)(\mathfrak{X}_1,\cdots,\mathfrak{X}_m,x_{m+1})\\
&=&\sum_{1\leq j<k\leq m}(-1)^{j}f(\mathfrak{X}_1,\cdots,\widehat{\mathfrak{X}_j},\cdots,\mathfrak{X}_{k-1},\mathfrak{X}_{j}\circ\mathfrak{X}_{k},\mathfrak{X}_{k+1},\cdots,\mathfrak{X}_m,x_{m+1})\\
&&+\sum_{j=1}^{m}(-1)^{j}f(\mathfrak{X}_1,\cdots,\widehat{\mathfrak{X}_j},\cdots,\mathfrak{X}_{m},[\mathfrak{X}_{j},x_{m+1}]_{R})\\
&&+\sum_{j=1}^{m}(-1)^{j+1}\varrho_{R}(\mathfrak{X}_{j})f(\mathfrak{X}_1,\cdots,\widehat{\mathfrak{X}_j},\cdots,\mathfrak{X}_m,x_{m+1})\\
&&+\sum_{i=1}^{n-1}(-1)^{n+m-i+1}\varrho_{R}(x^1_{m},\cdots,\widehat{x^i_{m}},\cdots,x^{n-1}_{m},x_{m+1})f(\mathfrak{X}_1,\cdots,\mathfrak{X}_{m-1},x^i_m),
\end{eqnarray*}
for all $\mathfrak{X}_i=x^1_i\wedge \cdots \wedge x^{n-1}_i\in \wedge^{n-1}\g, i=1,2,\cdots,m$ and $x_{m+1}\in \g.$

It is obvious that $f\in C_{\Ni}^{1}(\g;\g)$ is closed if and only if
\begin{eqnarray*}
&&f(\sum_{i=1}^{n}(-1)^{n-i}[Rx_1,\cdots,\widehat{Rx_{i}},\cdots,Rx_n,x_i]_{\g})\\
&=&f[Rx_1,\cdots,Rx_n]_{\g}+\sum_{i=1}^{n}(-1)^{n-i}[Rx_1,\cdots,\widehat{Rx_i},\cdots,Rx_n,f(x_i)]_{\g}\\
&&-R\Big(\sum_{i=1}^{n}\sum_{j=1,j\neq i}^{n}(-1)^{n-j}[Rx_1,\cdots,\widehat{Rx_j},\cdots,Rx_{i-1},f(x_i),Rx_{i+1},\cdots,Rx_n,x_j]_{\g}\Big)\\
&&+\sum_{i=1}^{n}(-1)^{n-i}R[Rx_1,\cdots,\widehat{Rx_i},\cdots,Rx_n,f(x_i)]_{\g},
\end{eqnarray*}
where $x_1,\cdots,x_n\in \g.$

For any $\mathfrak{X}\in \wedge^{n-1}\g, $ we define $\delta_{R}(\mathfrak{X}):\g\rightarrow\g$ by
\begin{eqnarray}\label{delta-define}
~\delta_{R}(\mathfrak{X})x=R[\mathfrak{X},x]_{\g}-[\mathfrak{X},Rx]_{\g}-R[\mathfrak{X},Rx]_{\g}, \quad \forall x\in \g.
\end{eqnarray}
\begin{pro}
Let $R$ be a Reynolds operator on an $n$-Lie algebra $(\g,[\cdot,\cdots,\cdot]_{\g}).$
Then $\delta_{R}(\mathfrak{X})$ is a $1$-cocycle on the $n$-Lie algebra $(\g,[\cdot,\cdots,\cdot]_{R})$ with coefficients in $(\g;\varrho_{R}).$
\end{pro}
\begin{proof}
By direct calculation we can get the conclusion.
\end{proof}

\begin{thm}\label{cohomology-Reynolds}
Let $R$ be a Reynolds operator on an $n$-Lie algebra $(\g,[\cdot,\cdots,\cdot]_{\g})$ with respect to a representation $(\g;\varrho_R).$ Define the set of $m$-cochains by
\begin{eqnarray*}
C_{R}^{m}(\g;\g)=
\left\{\begin{array}{rcl}
{}C_{\Ni}^{m}(\g;\g),\quad m\geq 1,\\
{}\wedge^{n-1}\g,\quad m=0.
\end{array}\right.
\end{eqnarray*}
Define ${\rm d}:C_{R}^{m}(\g;\g)\rightarrow C_{R}^{m+1}(\g;\g)$ by
\begin{eqnarray*}
{\rm d}=
\left\{\begin{array}{rcl}
{}{\rm d_{R}},\quad m\geq 1,\\
{}\delta_{R},\quad m=0,
\end{array}\right.
\end{eqnarray*}
where $\delta_{R}$ is given by \eqref{delta-define}.
Thus,  $(\oplus^{+\infty}_{m=0} C_{R}^{m}(\g;\g),{\rm d})$  is a cochain complex.
\end{thm}
\begin{defi}
The cohomology of the cochain complex $(\oplus^{+\infty}_{m=0} C_{R}^{m}(\g;\g),{\rm d})$ is defined to be the cohomology of the Reynolds operator $R.$
\end{defi}

Denote the set of $m$-cocycles by $\huaZ_{R}^m(\g;\g),$ the set of $m$-coboundaries by $\huaB_{R}^m(\g;\g)$ and $m$-th cohomology group for the Reynolds operator $R$ by
\begin{eqnarray}
\huaH_{R}^m(\g;\g)=\huaZ_{R}^m(\g;\g)/\huaB_{R}^m(\g;\g),\quad m\geq0.
\end{eqnarray}

We can use these cohomology groups to characterize infinitesimal deformations of Reynolds operators.

\subsection{Infinitesimal deformations of Reynolds operators on $n$-Lie algebras}
Next we use the cohomology theory to characterize the infinitesimal deformations of Reynolds operators.

Let $(\g,[\cdot,\cdots,\cdot]_{\g})$ be an $n$-Lie algebra over $\mathbb R$ and $\mathbb R[t]$ be the polynomial ring in one variable $t.$
Then $\mathbb R[t]/(t^2)\otimes_{\mathbb R}\g$ is a $\mathbb R[t]/(t^2)$-module, moreover, $\mathbb R[t]/(t^2)\otimes_{\mathbb R}\g$ is an $n$-Lie algebra over $\mathbb R[t]/(t^2)$, where the $n$-Lie algebra structure is defined by
\begin{eqnarray*}
[f_1(t)\otimes_{\mathbb R} x_1,\cdots,f_n(t)\otimes_{\mathbb R} x_n]_{t}= f_1(t)\cdots f_n(t)\otimes_{\mathbb R}[x_1,\cdots,x_n]_{\g},
\end{eqnarray*}
for $f_{i}(t)\in \mathbb R[t]/(t^2),1\leq i\leq n,x_1,\cdots,x_n\in \g.$

In the sequel, all the vector spaces are finite dimensional vector spaces over $\mathbb R$ and we denote $f(t)\otimes_{\mathbb R} x$ by $f(t)x,$ where $f(t)\in \mathbb R[t]/(t^2).$

\begin{defi}
Let $R:\g\rightarrow \g$ be a Reynolds operator on an $n$-Lie algebra $(\g,[\cdot,\cdots,\cdot]_{\g})$ and
$\frkR:\g\rightarrow\g$ a linear map. If $R_t=R+t\frkR$ is still a Reynolds operator on the $n$-Lie algebra
$(\mathbb R[t]/(t^2)\otimes_{\mathbb R}\g,[\cdot,\cdots,\cdot]_{t})$, we say that $\frkR$ generates an {\bf infinitesimal  deformation} of the Reynolds operator $R$.
\end{defi}
Since $R_t=R+t\frkR$ is a Reynolds operator on the $n$-Lie algebra $(\mathbb R[t]/(t^2)\otimes_{\mathbb R}\g,[\cdot,\cdots,\cdot]_{t})$,
for any $x_1,\cdots,x_n\in \g,$
we have
\begin{eqnarray*}
 [R_t(x_1),\cdots,R_t(x_n)]_{t}=\sum^n_{i=1}(-1)^{n-i} R_t[R_t(x_1),\cdots,\widehat{R_t( x_i)},\cdots,R_t(x_n),x_i]_{t}-R_t[R_t(x_1),\cdots,R_t(x_n)]_{t},
\end{eqnarray*}
which implies that
\begin{eqnarray}
\label{equivalent-1}\qquad&&[\frkR x_1,Rx_2,\cdots,Rx_n]_{\g}+\cdots+[Rx_1,Rx_2,\cdots,\frkR x_n]_{\g}\\
\nonumber &=&\sum_{i=1}^{n}\frkR[Rx_1,\cdots,x_i,\cdots,Rx_n]_{\g}-\frkR[Rx_1,\cdots,Rx_n]_{\g}-\sum_{i=1}^{n}R[Rx_1,\cdots,\frkR x_i,\cdots,Rx_n]_{\g}\\
\nonumber&&+\sum_{i=1}^{n}\sum_{j=1,j\neq i}^{n}R[Rx_1,\cdots,x_i,\cdots,Rx_{i-1},\frkR(x_j),Rx_{i+1},\cdots,Rx_n]_{\g}.
\end{eqnarray}

 Note that \eqref{equivalent-1} means that $\frkR$ is a $1$-cocycle of the $n$-Lie algebra $(\g,[\cdot,\cdots,\cdot]_{R})$ with coefficients in $\g$, that is, $\rm d{\frkR}=0.$

\begin{defi}
Let $R$ be a Reynolds operator on an $n$-Lie algebra $\g$. Two infinitesimal deformations $R^1_{t}= R+t\frkR_{1}$ and  $R^2_{t}=R+t\frkR_{2}$ are said to be {\bf equivalent} if there exist $\mathfrak{X}\in\wedge^{n-1}\g$ , such that the pair
$$(\phi=\Id_{\g}+t \ad_{\mathfrak{X}},\psi=\Id_{\g}+t\ad_{\mathfrak{X}}-t[\mathfrak{X},R])$$
is a homomorphism from $R^1_{t}$ to $R^2_{t}$, where $[\mathfrak{X},R]:\g\rightarrow\g$ is defined by
\begin{eqnarray}
[\mathfrak{X},R](x)=[\mathfrak{X},Rx]_{\g},\quad \forall x\in \g.
\end{eqnarray}
 In particular,
an infinitesimal deformation $R_{t}=R+t\frkR_{1}$ of a Reynolds operator $R$ is said to be {\bf trivial} if there exist $\mathfrak{X}\in \wedge^{n-1}\g$ such that $(\Id_{\g}+t\ad_{\mathfrak{X}},\Id_{\g}+t\ad_{\mathfrak{X}}-t[\mathfrak{X},R])$  is a homomorphism from $R_{t}$ to $R.$
\end{defi}
\begin{thm}
Let $R$ be a Reynolds operator on an $n$-Lie algebra $\g.$
If two infinitesimal deformations $R^1_{t}= R+t\frkR_{1}$ and  $R^2_{t}=R+t\frkR_{2}$ of $R$ are equivalent, then
$\frkR_{1}$ and $\frkR_{2}$ are in the same cohomology class.
\end{thm}
\begin{proof}
Let $(\phi,\psi)$ be two linear maps, where $\phi=\Id_{\g}+t\ad_{\mathfrak{X}}$ and $\psi=\Id_{\g}+t\ad_{\mathfrak{X}}-t[\mathfrak{X},R]$,
which give an equivalence between two infinitesimal deformations $R^1_{t}=R+t\frkR_{1}$ and  $R^2_{t}=R+t\frkR_{2}$ of the Reynolds operator $R$. By \eqref{condition-1}, we have
\begin{eqnarray*}
 \frkR_1(x)&=&\frkR_2(x)+R[\mathfrak{X},x]_{\g}-[\mathfrak{X},Rx]_{\g}-R[\mathfrak{X},Rx]_{\g}\\
          &=& \frkR_2(x)+({\rm d}\mathfrak{X})(x),\quad \forall x\in \g,
\end{eqnarray*}
which implies that $\frkR_1$ and $\frkR_2$ are in the same cohomology class.
\end{proof}

\section{NS-$n$-Lie algebras, Reynolds operators and Nijenhuis operators}\label{sec:GM}
In this section,  we introduce the notion of an NS-$n$-Lie algebra and show that  both Reynolds operators and Nijenhuis operators  on an $n$-Lie algebra   induce   NS-$n$-Lie algebras.

\subsection{NS-$n$-Lie algebras and Reynolds operators on $n$-Lie algebras }

First we introduce the notion of an NS-$n$-Lie algebra, which reduces to an NS-Lie algebra introduced in \cite{Das-1} when $n=2.$
\begin{defi}\label{defi-NS-n-Lie algebra}
Let $A$ be a vector space together with two multilinear maps $\{\cdot,\cdots,\cdot\}:\wedge^ {n-1}A\otimes A\rightarrow A$ and $[\cdot,\cdots,\cdot]:\wedge^{n}A\rightarrow A$. The tuple $(A,\{\cdot,\cdots,\cdot\},[\cdot,\cdots,\cdot])$ is called an
{\bf NS-$n$-Lie algebra} if the following identities hold:
\begin{eqnarray}
\label{NS-n-Lie-1}\{x_1,\cdots,x_{n-1},\{y_1,\cdots,y_{n-1},y_n\}\}&=&\{y_1,\cdots,y_{n-1},\{x_1,\cdots,x_{n-1},y_n\}\}\\
\nonumber&&+\sum\limits^{n-1}_{j=1}\big\{y_1,\cdots,y_{j-1},\Courant{x_1,\cdots,x_{n-1},y_j},y_{j+1},\cdots,y_n\big\}\\
\label{NS-n-Lie-2}\{\Courant{y_1,\cdots,y_n},x_1,\cdots,x_{n-1}\}&=&\sum\limits^{n}_{j=1}(-1)^{n-j}\{y_1,\cdots,\widehat{y_{j}},\cdots,y_n,\{y_{j},x_1,\cdots,x_{n-1}\}\},\\
\label{NS-n-Lie-3}\big[x_1,\cdots,x_{n-1},\Courant{y_1,\cdots,y_{n-1},y_n}\big]&=&\sum\limits^{n}_{j=1}(-1)^{n-j}\big[ y_1,\cdots,\widehat{y_j},\cdots,y_n,\Courant{x_1,\cdots,x_{n-1},y_j}\big]\\
\nonumber&&-\{x_1,\cdots,x_{n-1},[y_1,\cdots,y_n]\}\\
\nonumber&&+\sum\limits^{n}_{j=1}(-1)^{n-j}\{y_1,\cdots,\widehat{y_j},\cdots,y_n,[x_1,\cdots,x_{n-1},y_j]\},
\end{eqnarray}
where $x_i,y_{j}\in A, 1\leq i\leq n-1,1\leq j\leq n$ and $\Courant{\cdot,\cdots,\cdot}:\wedge^{n}\g\rightarrow\g$ is a multilinear map defined by
\begin{eqnarray}
\label{NS-n-Lie-4} \Courant{y_1,\cdots,y_n}=\sum\limits^{n}_{i=1}(-1)^{n-i}\{y_1,\cdots,\widehat{y_i},\cdots,y_n,y_i\}+[y_1,\cdots,y_n].
\end{eqnarray}
\end{defi}

\begin{rmk}
Let $(A,\{\cdot,\cdots,\cdot\},[\cdot,\cdots,\cdot])$  be an NS-$n$-Lie algebra. On the one hand, if $\{\cdot,\cdots,\cdot\}=0,$
we get that $(A,[\cdot,\cdots,\cdot])$ is an $n$-Lie algebra. On the other hand, if $[\cdot,\cdots,\cdot]=0,$ then $(A,\{\cdot,\cdots,\cdot\})$
is an $n$-pre-Lie algebra which was introduced in \cite{Liu-Ma}.
Thus, NS-$n$-Lie algebras are  generalizations of both $n$-Lie algebras and $n$-pre-Lie algebras. Moreover, an NS-$n$-Lie algebra reduces to an NS-Lie algebra introduced in \cite{Das-1} when $n=2.$
\end{rmk}

\begin{thm}
Let $(A,\{\cdot,\cdots,\cdot\},[\cdot,\cdots,\cdot])$ be an NS-$n$-Lie algebra. Then $(A, \Courant{\cdot,\cdots,\cdot})$
is an $n$-Lie algebra, which is called the {\bf sub-adjacent $n$-Lie algebra} of $(A,\{\cdot,\cdots,\cdot\},[\cdot,\cdots,\cdot])$, and denoted by $A^c$.
Moreover, define a skew-symmetric multilinear map $L:\otimes^{n-1} A\rightarrow \gl(A)$ by
$$ L(x_1,\cdots,x_{n-1})(x_n)=\{x_1,\cdots,x_{n-1},x_n\},\quad \forall~x_1,\cdots,x_{n-1}\in A^c,x_n\in A.$$
Then $(A;L)$ is a representation of the $n$-Lie algebra $A^c.$
\end{thm}

\begin{proof}
For $x_1,\cdots,x_{n-1},y_1,\cdots,y_{n}\in A,$ by \eqref{NS-n-Lie-1}, \eqref{NS-n-Lie-2} and \eqref{NS-n-Lie-3}, we have
\begin{eqnarray*}
&&\Courant{x_1,\cdots,x_{n-1},\Courant{y_1,\cdots,y_n}}-\sum\limits^{n}_{i=1}\Courant{y_1,\cdots,y_{i-1},\Courant{x_1,\cdots,x_{n-1},y_{i}},y_{i+1},\cdots,y_n}\\
&=&\sum\limits^{n-1}_{i=1}(-1)^{n-i}\{x_1,\cdots,\widehat{x_i},\cdots,x_{n-1},\Courant{y_1,\cdots,y_n},x_i\}+[x_1,\cdots,x_{n-1},\Courant{y_1,\cdots,y_n}]\\
&&+\{x_1,\cdots,x_{n-1},\sum\limits^{n}_{i=1}(-1)^{n-i}\{y_1,\cdots,\widehat{y_i},\cdots,y_n,y_i\}+[y_1,\cdots,y_n]\}\\
&&-\sum\limits^{n}_{i=1}(-1)^{n-i}\{y_1,\cdots,y_{i-1},\widehat{y_{i}},y_{i+1},\cdots,y_n,\sum\limits^{n-1}_{j=1}(-1)^{n-j}\{x_1,\cdots,\widehat{x_j},\cdots,x_{n-1},y_i,x_j\}\\
&&+\{x_1,\cdots,x_{n-1},y_i\}+[x_1,\cdots,x_{n-1},y_i]\}-\sum\limits^{n}_{i=1}(-1)^{n-i}[y_1,\cdots,\widehat{y_{i}},\cdots,y_n,\Courant{x_1,\cdots,x_{n-1},y_i}]\\
&&-\sum\limits^{n}_{i=1}\sum\limits^{n}_{k=1,k\neq i}(-1)^{n-k}\{y_1,\cdots,\widehat{y_k},\cdots,y_{i-1},\Courant{x_1,\cdots,x_{n-1},y_i},y_{i+1},\cdots,y_n,y_k\}\\
&=&\sum\limits^{n}_{i=1}\{x_1,\cdots,x_{n-1},\{y_1,\cdots,y_i,\cdots,y_n\}\}-\sum\limits^{n}_{i=1}(-1)^{n-i}\{y_1,\cdots,\widehat{y_{i}},\cdots,y_n,\{x_1,\cdots,x_{n-1},y_i\}\}\\
&&-\sum\limits^{n}_{i=1}\sum\limits^{n}_{k=1,k\neq i}(-1)^{n-k}\{y_1,\cdots,\widehat{y_k},\cdots,y_{i-1},\Courant{x_1,\cdots,x_{n-1},y_i},y_{i+1},\cdots,y_n,y_k\}\\
&&+\sum\limits^{n-1}_{i=1}\{x_1,\cdots,x_i,\cdots,x_{n-1},\Courant{y_1,\cdots,y_n}\}-\sum\limits^{n}_{i=1}(-1)^{n-i}\{y_1,\cdots,\widehat{y_{i}},\cdots,y_n,[x_1,\cdots,x_{n-1},y_i]\}\\
&&-\sum\limits^{n}_{i=1}\sum\limits^{n-1}_{j=1}(-1)^{i+j}\{y_1,\cdots,y_{i-1},\widehat{y_{i}},y_{i+1},\cdots,y_n,\{x_1,\cdots,\widehat{x_j},\cdots,x_{n-1},y_i,x_j\}\}\\
&&+[x_1,\cdots,x_{n-1},\Courant{y_1,\cdots,y_n}]+\{x_1,\cdots,x_{n-1},[y_1,\cdots,y_n]\}\\
&&-\sum\limits^{n}_{i=1}(-1)^{n-i}[y_1,\cdots,\widehat{y_{i}},\cdots,y_n,\Courant{x_1,\cdots,x_{n-1},y_i}]\\
&=&0,
\end{eqnarray*}
which implies that $(A, \Courant{\cdot,\cdots,\cdot})$ is an $n$-Lie algebra.

By \eqref{NS-n-Lie-1}, we have
 \begin{eqnarray*}
&&[L(\mathfrak{X}),L(\mathfrak{Y})](x_n)-L(\mathfrak{X}\circ\mathfrak{Y})(x_n)\\
&=&\{x_1,\cdots,x_{n-1},\{y_1,\cdots,y_{n-1},x_n\}\}-\{y_1,\cdots,y_{n-1},\{x_1,\cdots,x_{n-1},x_n\}\}\\
&&-\{\sum\limits^{n-1}_{i=1}y_1,\cdots,y_{i-1},\Courant{x_1,\cdots,x_{n-1},y_i},y_{i+1},\cdots,y_{n-1},x_n\}\\
&=&0,
 \end{eqnarray*}
which implies that \eqref{n-representation-1} in Definition \ref{defi-representation} holds, where $\mathfrak{X}=x_1\wedge \cdots\wedge x_{n-1}$, $\mathfrak{Y}=y_1\wedge \cdots\wedge y_{n-1}$.

By \eqref{NS-n-Lie-2}, we have
 \begin{eqnarray*}
&&L(x_1,\cdots,x_{n-2},\Courant{y_1,\cdots,y_n})(x_{n-1})-\sum\limits^{n}_{i=1}(-1)^{n-i}L(y_1,\cdots,\widehat{y_i},\cdots,y_n)L(x_1,\cdots,x_{n-2},y_i)(x_{n-1})\\
&=&\{\Courant{y_1,\cdots,y_n},x_1,\cdots,x_{n-2},x_{n-1}\}-\sum\limits^{n}_{i=1}(-1)^{n-i}\{y_1,\cdots,\widehat{y_i},\cdots,y_n,\{y_i,x_1,\cdots,x_{n-2},x_{n-1}\}\\
&=&0,
 \end{eqnarray*}
which implies that \eqref{n-representation-2} in Definition \ref{defi-representation} holds.
Therefore, $(A;L)$ is a representation of the sub-adjacent  $n$-Lie algebra $A^c.$
\end{proof}

\begin{thm}\label{Reynolds-NS-n-algebra}
Let $R:\g\rightarrow\g$ be a Reynolds operator on an $n$-Lie algebra $(\g,[\cdot,\cdots,\cdot]_{\g}).$
Then $(\g,\{\cdot,\cdots,\cdot\},[\cdot,\cdots,\cdot])$ is an NS-$n$-Lie algebra, where $\{\cdot,\cdots,\cdot\}$ and $[\cdot,\cdots,\cdot]$ are defined by
 \begin{eqnarray}
~\{x_1,\cdots,x_{n-1},x_n\}&=&[Rx_1,\cdots,Rx_{n-1},x_n]_{\g},\\
~ [x_1,\cdots,x_{n-1},x_n]&=&-[Rx_1,\cdots,Rx_{n}]_{\g}, \quad \forall x_1,\cdots,x_n\in\g.
 \end{eqnarray}

\end{thm}
 \begin{proof}
 For $x_1,\cdots,x_{n-1},y_1,\cdots,y_{n}\in \g$, by \eqref{FI-Identity}, \eqref{n-Reynolds} and \eqref{NS-n-Lie-4}, we have
  \begin{eqnarray*}
&&\{x_1,\cdots,x_{n-1},\{y_1,\cdots,y_{n-1},y_n\}\}-\{y_1,\cdots,y_{n-1},\{x_1,\cdots,x_{n-1},y_n\}\}\\
&&-\sum\limits^{n-1}_{j=1}\{y_1,\cdots,y_{j-1},\Courant{x_1,\cdots,x_{n-1},y_j},y_{j+1},\cdots,y_n\}\\
&=&[Rx_1,\cdots,Rx_{n-1},[Ry_1,\cdots,Ry_{n-1},y_n]_{\g}]_{\g}-[Ry_1,\cdots,Ry_{n-1},[Rx_1,\cdots,Rx_{n-1},y_n]_{\g}]_{\g}\\
&&-\sum\limits^{n-1}_{j=1}[Ry_1,\cdots,Ry_{j-1},[Rx_1,\cdots,Rx_{n-1},Ry_j]_{\g},Ry_{j+1},\cdots,y_n]_{\g}\\
&=&0.
\end{eqnarray*}
This implies that \eqref{NS-n-Lie-1} in Definition \ref{defi-NS-n-Lie algebra} holds.
 Similarly, we can verify that \eqref{NS-n-Lie-2} and \eqref{NS-n-Lie-3} hold.
 \emptycomment{\begin{eqnarray*}
 &&\{\Courant{y_1,\cdots,y_n},x_1,\cdots,x_{n-1}\}-\sum\limits^{n}_{j=1}(-1)^{n-j}\{y_1,\cdots,\widehat{y_{j}},\cdots,y_n,\{y_{j},x_1,\cdots,x_{n-1}\}\}\\
 &=&[R[Ry_1,\cdots,Ry_n]_{\g},Rx_1,\cdots,Rx_{n-2},x_{n-1}]_{\g}\\
 &&-\sum\limits^{n}_{j=1}(-1)^{n-j}[Ry_1,\cdots,\widehat{Ry_{j}},\cdots,Ry_n,[Ry_j,Rx_1,\cdots,Rx_{n-2},x_{n-1}]_{\g}]_{\g}\\
 &=&0.
 \end{eqnarray*}
 Moreover, we observe that
 \begin{eqnarray*}
 &&[x_1,\cdots,x_{n-1},\Courant{y_1,\cdots,y_{n-1},y_n}]-\sum\limits^{n}_{j=1}(-1)^{n-j}[ y_1,\cdots,\widehat{y_j},\cdots,y_n,\Courant{x_1,\cdots,x_{n-1},y_j}]\\
&&+\{x_1,\cdots,x_{n-1},[y_1,\cdots,y_n]\}-\sum\limits^{n}_{j=1}(-1)^{n-j}\{y_1,\cdots,\widehat{y_j},\cdots,y_n,[x_1,\cdots,x_{n-1},y_j]\}\\
 &=&
  \end{eqnarray*}}
Thus $(\g,\{\cdot,\cdots,\cdot\},[\cdot,\cdots,\cdot])$ is an NS-$n$-Lie algebra.
\end{proof}

\emptycomment{

\begin{pro}
Let $R$ and $R'$ be Reynolds operators on an $n$-Lie algebra $\g$, and $(\g,\{\cdot,\cdot,\cdot\},[\cdot,\cdot,\cdot])$, $(\g,\{\cdot,\cdot,\cdot\}',[\cdot,\cdot,\cdot]')$ be the induced NS-$n$-Lie algebras.  Let $(\phi,\psi)$ be a homomorphism  from $R$ to $R'$. Then $\psi$ is a homomorphism from the NS-$n$-Lie algebra $(\g,\{\cdot,\cdot,\cdot\},[\cdot,\cdot,\cdot])$ to $(\g,\{\cdot,\cdot,\cdot\}',[\cdot,\cdot,\cdot]')$.
\end{pro}
\begin{proof}
For all $u,v,w\in V,$ by \eqref{condition-1}-\eqref{condition-3} and \eqref{twisted-Ns}, we have
\begin{eqnarray*}
\psi(\{u,v,w\})&=&\psi(\rho(Tu,Tv)w)=\rho(\phi(Tu),\phi(Tv))\psi(w)=\rho(T'\psi(u),T'\psi(v))\psi(w)\\
&=&\{\psi(u),\psi(v),\psi(w)\}',\\
 \psi([u,v,w])&=&\psi(\Phi(Tu,Tv,Tw))=\Phi(\phi(Tu),\phi(Tv),\phi(Tw))=\Phi(T'\psi(u),T'\psi(v), T'\psi(w))\\
&=&[\psi(u),\psi(v),\psi(w)]',
\end{eqnarray*}
which implies that $\psi$ is a homomorphism between the induced NS-3-Lie algebras.
\end{proof}

Thus, Theorem \ref{construct-3-NS-Lie algebra} can be enhanced to a functor from the category of $\Phi$-twisted Rota-Baxter operators on a $3$-Lie algebra $\g$ with respect to a representation $(V;\rho)$ to the category of NS-3-Lie algebras.
}
\subsection{NS-$n$-Lie algebras and Nijenhuis operators on $n$-Lie algebras}
In this subsection, we show that a Nijenhuis operator on an $n$-Lie algebra gives rise to an NS-$n$-Lie algebra. First we recall some properties of Nijenhuis operators from \cite{Liu-Jie-Feng}.

Let $(\g,[\cdot,\cdots,\cdot]_{\g})$ be an $n$-Lie algebra, and $N:\g\rightarrow\g$ a linear map. Define an $n$-ary bracket $[\cdot,\cdots,\cdot]^{1}_{N}:\wedge^{n}\g\rightarrow\g$ by
\begin{eqnarray}\label{rep-Nijenhuis-1}
[x_1,x_2,\cdots,x_n]^{1}_{N}=\sum\limits^{n}_{i=1}[x_1,\cdots,Nx_i,\cdots,x_n]_{\g}-N[x_1,x_2,\cdots,x_n]_{\g}.
\end{eqnarray}
 Define the $n$-ary brackets $[\cdot,\cdots,\cdot]_N^j:\wedge^n\g\longrightarrow\g, (2\leq j\leq n-1)$ via induction by
    \begin{equation}\label{rep-Nijenhuis-2}
    [x_1,x_2,\cdots,x_n]_N^{j}=\sum_{i_1<i_2\cdots<i_{j}}[\cdots,Nx_{i_1},\cdots,Nx_{i_{j}},\cdots]_{\g}-N[x_1,x_2,\cdots,x_n]_N^{j-1}.
  \end{equation}
  In particular, we have
    \begin{equation}\label{rep-Nijenhuis-3}
    [x_1,x_2,\cdots,x_n]_N^{n-1}=\sum_{i_1<i_2\cdots<i_{n-1}}[\cdots,Nx_{i_1},\cdots,Nx_{i_{n-1}},\cdots]_{\g}-N[x_1,x_2,\cdots,x_n]_N^{n-2}.
  \end{equation}
\begin{defi}\cite{Liu-Jie-Feng}
Let $(\g,[\cdot,\cdots,\cdot]_{\g})$ be an $n$-Lie algebra. A linear map $N:\g\rightarrow\g$ is called a {\bf Nijenhuis operator} if
\begin{eqnarray}\label{Nijenhuis-1}
[Nx_1,Nx_2,\cdots,Nx_n]_{\g}=N([x_1,x_2,\cdots,x_n]^{n-1}_{N}),\quad \forall x_1,\cdots,x_n\in \g.
\end{eqnarray}
\end{defi}
In this case, the $n$-Lie algebra $\g$ carries a new $n$-Lie bracket $[\cdot,\cdots,\cdot]^{n-1}_{N}$, the $n$-Lie algebra $(\g,[\cdot,\cdots,\cdot]^{n-1}_{N})$ will be called the {\bf deformed $n$-Lie algebra} and denote this $n$-Lie algebra  by $\g_{N}.$ It is obvious that  $N$ is a homomorphism from  the deformed $n$-Lie algebra $(\g,[\cdot,\cdots,\cdot]^{n-1}_{N})$ to $(\g,[\cdot,\cdots,\cdot]_{\g})$.

\begin{lem}
Let $N$ be a Nijenhuis operator on an $n$-Lie algebra $(\g,[\cdot,\cdots,\cdot]_{\g}).$ Define $\rho_{N}:\wedge^{n-1}\g_{N}\rightarrow\gl(\g)$ by
\begin{eqnarray}\label{Nijenhuis-representation}
\rho_{N}(x_1,\cdots,x_{n-1})x=[Nx_1,\cdots,Nx_{n-1},x]_{\g},\quad \forall x_1,\cdots,x_{n-1}\in \g_{N}, x\in \g.
\end{eqnarray}
 Then $(\g;\rho_{N})$ is a representation of the deformed $n$-Lie algebra $(\g,[\cdot,\cdots,\cdot]^{n-1}_{N}).$
\end{lem}

\begin{proof}
For all $x_1,\cdots,x_{n-1}\in \g_{N},x\in \g,$ by \eqref{FI-Identity}, \eqref{rep-Nijenhuis-1}-\eqref{rep-Nijenhuis-3}, we have
\begin{eqnarray*}
&&[\rho_{N}(x_1,\cdots,x_{n-1}),\rho_{N}(y_1,\cdots,y_{n-1})](x)\\
&&\quad-\rho_{N}\Big(\sum\limits^{n-1}_{i=1}(y_1,\cdots,y_{i-1},[x_1,\cdots,x_{n-1},y_i]^{n-1}_{N},y_{i+1},\cdots,y_{n-1})\Big)(x)\\
&=&[Nx_1,\cdots,Nx_{n-1},[Ny_1,\cdots,Ny_{n-1},x]_{\g}]_{\g}-[Ny_1,\cdots,Ny_{n-1},[Nx_1,\cdots,Nx_{n-1},x]_{\g}]_{\g}\\
&&-\sum\limits^{n-1}_{j=1}[Ny_1,\cdots,Ny_{j-1},[Nx_1,\cdots,Nx_{n-1},Ny_j]_{\g},Ny_{j+1},\cdots,x]_{\g}\\
&=&0,
\end{eqnarray*}
which implies that \eqref{n-representation-1} in Definition \ref{defi-representation} holds.
 Similarly, we can show that \eqref{n-representation-2} also holds. Thus, $(\g;\rho_{N})$ is a representation of the deformed $n$-Lie algebra $(\g,[\cdot,\cdots,\cdot]^{n-1}_{N}).$
\end{proof}

At the end of this section, we show that a Nijenhuis operator induces an NS-$n$-Lie algebra.
\begin{thm}
Let $(\g,[\cdot,\cdots,\cdot]_{\g})$ be an $n$-Lie algebra and $N:\g\rightarrow\g$ be a Nijenhuis operator. Then there exists an NS-$n$-Lie algebra
on $\g$ given by
\begin{eqnarray}\label{Nijenhuis-NS}
~\{x_1,\cdots,x_n\}&=&[Nx_1,\cdots,Nx_{n-1},x_n]_{\g},\\
~ [x_1,\cdots,x_n]&=&-N[x_1,x_2,\cdots,x_n]^{n-2}_{N},\quad \forall x_1,\cdots,x_n\in \g.
\end{eqnarray}
\end{thm}
\begin{proof}
For any $x_1,\cdots,x_n\in \g,$ by \eqref{NS-n-Lie-4} and \eqref{rep-Nijenhuis-2}, we have
 \begin{eqnarray*}
 \Courant{x_1,\cdots,x_n}&=&\sum\limits^{n}_{i=1}(-1)^{n-i}\{x_1,\cdots,\widehat{x_i},\cdots,x_n,x_i\}+[x_1,\cdots,x_n]\\
 &=&\sum\limits_{i_1<i_2\cdots<i_{n-1}}[\cdots,Nx_{i_{1}},\cdots,Nx_{i_{n-1}},\cdots]_{\g}-N[x_1,x_2,\cdots,x_n]^{n-2}_{N}\\
  &=&[x_1,\cdots,x_n]^{n-1}_{N}.
\end{eqnarray*}
Then by \eqref{FI-Identity} and \eqref{n-Reynolds}, we have
 \begin{eqnarray*}
&&\{x_1,\cdots,x_{n-1},\{y_1,\cdots,y_{n-1},y_n\}\}-\{y_1,\cdots,y_{n-1},\{x_1,\cdots,x_{n-1},y_n\}\}\\
&&-\sum\limits^{n-1}_{j=1}\{y_1,\cdots,y_{j-1},\Courant{x_1,\cdots,x_{n-1},y_j},y_{j+1},\cdots,y_n\}\\
&=&[Nx_1,\cdots,Nx_{n-1},[Ny_1,\cdots,Ny_{n-1},y_n]_{\g}]_{\g}-[Ny_1,\cdots,Ny_{n-1},[Nx_1,\cdots,Nx_{n-1},y_n]_{\g}]_{\g}\\
&&-\sum\limits^{n-1}_{j=1}[Ny_1,\cdots,Ny_{j-1},[Nx_1,\cdots,Nx_{n-1},Ny_j]_{\g},Ny_{j+1},\cdots,y_n]_{\g}\\
&=&0.
\end{eqnarray*}
This implies that \eqref{NS-n-Lie-1} in Definition \ref{defi-NS-n-Lie algebra} holds.
 Similarly, we can verify that \eqref{NS-n-Lie-2} and \eqref{NS-n-Lie-3} hold. Hence the conclusion follows.
\end{proof}

\section{Constructions of Reynolds   $n$-Lie algebras}\label{sec:Con}

\subsection{Constructions of Reynolds $(n+1)$-Lie algebras from Reynolds $n$-Lie algebras }

In \cite{BRP-2}, the authors constructed an $(n+1)$-Lie algebra $\g_{f}$ from an $n$-Lie algebra $\g$
using a linear function $f\in \g^*$ . In this section, we provide the condition for a Reynolds operator on an $n$-Lie algebra $\g$ also being a Reynolds operator on the $(n+1)$-Lie algebra $\g_{f}$.
\begin{lem}\cite{BRP-2}\label{constrcuted-alg}
Let $(\g,[\cdot,\cdots,\cdot]_{\g})$ be an $n$-Lie algebra and $\g^*$ the dual space of $\g.$ Suppose that $f\in \g^*$ satisfies $f([x_1,\cdots,x_n]_{\g})=0$ for all $x_i\in\g.$ Then $(\g,\{\cdot,\cdots,\cdot\})$ is an $(n+1)$-Lie algebra, where the bracket $\{\cdot,\cdots,\cdot\}:\wedge^{n+1}\g\rightarrow\g$ is given by
\begin{eqnarray}\label{constrcuted-n+1-Alge}
\{x_1,\cdots,x_{n+1}\}=\sum\limits^{n+1}_{i=1}(-1)^{i-1}f(x_i)[x_1,\cdots,\widehat{x_{i}},\cdots,x_{n+1}]_{\g},\quad \forall x_i\in \g, 1\leq i\leq n+1.
\end{eqnarray}
The $(n+1)$-Lie algebra constructed as above is denoted by $\g_{f}.$
\end{lem}

Next we combine this result with Reynolds operators.
\begin{thm}\label{n+1Reynold}
Let $(\g,[\cdot,\cdots,\cdot]_{\g},R)$ be a Reynolds $n$-Lie algebra. If $f\in \g^*$ satisfy $f([x_1,\cdots,x_n]_{\g})=0,$ for all $x_1,\cdots,x_n\in \g,$ then $R$ is a Reynolds operator on the $(n+1)$-Lie algebra $\g_f$ defined by \eqref{constrcuted-n+1-Alge} if and only if $R$ satisfies
\begin{eqnarray}\label{n+1Reynold-equation}
\sum\limits^{n+1}_{i=1}(-1)^{n+1-i}f(x_i)R[Rx_i,\cdots,\widehat{Rx_i},\cdots,Rx_{n+1}]_{\g}=0.
\end{eqnarray}
\end{thm}
\begin{proof}
By \eqref{n-Reynolds} and \eqref{constrcuted-n+1-Alge}, for all $x_1,\cdots,x_{n+1}\in \g,$ we have
\begin{eqnarray*}
&&\sum\limits^{n+1}_{i=1}(-1)^{n+1-i}R\{Rx_i,\cdots,\widehat{Rx_i},\cdots,Rx_{n+1},x_i\}-R\{Rx_1,\cdots,Rx_{n+1}\}-\{Rx_1,\cdots,Rx_{n+1}\}\\
&=&\sum\limits^{n+1}_{i=1}(-1)^{n+1-i}R\Big(\sum\limits^{i-1}_{j=1}(-1)^{j-1}f(Rx_j)[Rx_1,\cdots,\widehat{Rx_j},\cdots,\widehat{Rx_i},\cdots,Rx_{n+1},x_i]_{\g}\\
&&+\sum\limits^{n+1}_{j=i+1}(-1)^{j}f(Rx_j)[Rx_1,\cdots,\widehat{Rx_i},\cdots,\widehat{Rx_j},\cdots,Rx_{n+1},x_i]_{\g}+f(x_i)[Rx_1,\cdots,Rx_i,\cdots,Rx_{n+1}]_{\g}\Big)\\
&&-\sum\limits^{n+1}_{i=1}f(Rx_i)(-1)^{i-1}R[Rx_1,\cdots,\widehat{Rx_i},\cdots,Rx_{n+1}]_{\g}+\sum\limits^{n+1}_{i=1}f(Rx_i)(-1)^{i-1}R[Rx_1,\cdots,\widehat{Rx_i},\cdots,Rx_{n+1}]_{\g}\\
&&-\sum\limits^{n+1}_{i=1}f(Rx_i)(-1)^{i-1}\Big(\sum\limits^{i-1}_{j=1}(-1)^{n-j}R[Rx_1,\cdots,\widehat{Rx_j},\cdots,\widehat{Rx_i},\cdots,Rx_{n+1},x_j]_{\g}\\
&&+\sum\limits^{n+1}_{j=i+1}(-1)^{n+1-j}R[Rx_1,\cdots,\widehat{Rx_i},\cdots,\widehat{Rx_j},\cdots,Rx_{n+1},x_j]_{\g}\Big)\\
&=&\sum\limits^{n+1}_{i=1}(-1)^{n+1-i}f(x_i)R[Rx_i,\cdots,\widehat{Rx_i},\cdots,Rx_{n+1}]_{\g}.
\end{eqnarray*}
Therefore, $R$ is a Reynolds operator on the $(n+1)$-Lie algebra $(\g,\{\cdot,\cdots,\cdot\})$ if and only if \eqref{n+1Reynold-equation} holds.
\end{proof}

\begin{cor}\label{3-Lie-Reynold-algebra}
Let $(\g,[\cdot,\cdots,\cdot]_{\g},R)$ be a Reynolds $n$-Lie algebra. If $R$ is a Reynolds operator on the $(n+1)$-Lie algebra $\g_f$ defined by \eqref{constrcuted-n+1-Alge}, then $(\g,\{\cdot,\cdots,\cdot\}_{R},R)$ is a Reynolds $(n+1)$-Lie algebra, where $\{\cdot,\cdots,\cdot\}_{R}$ is defined by
\begin{eqnarray*}
\{x_1,x_2,\cdots,x_{n+1}\}_{R}&=&\sum\limits^{n+1}_{i=1}\sum\limits^{n+1}_{j=1,j\neq i}(-1)^{n-i+j}f(Rx_j)[Rx_1,\cdots,\widehat{Rx_{i}},\cdots,\widehat{Rx_{j}},\cdots,Rx_{n+1},x_i]_{\g}\\
&&+\sum\limits^{n+1}_{j=1}(-1)^{j-1}f(Rx_j)[Rx_{1},\cdots,\widehat{Rx_{j}},\cdots,Rx_{n+1}]_{\g}.
\end{eqnarray*}
\end{cor}

\begin{ex}{\rm
Let $(\g,[\cdot,\cdot])$ be the $3$-dimensional Lie algebra given by
\begin{eqnarray*}
[e_1,e_2]=e_2,
\end{eqnarray*}
 where $\{e_1,e_2,e_3\}$ is a basis of $\g.$ By Lemma \ref{constrcuted-alg}, the trace function $f\in \g^*,$ where
$
\left\{\begin{array}{rcl}
{}f(e_1)=1,\\
{}f(e_2)=0,\\
{}f(e_3)=1,
\end{array}\right.
$
induces a $3$-Lie algebra $(\g_{f},\{\cdot,\cdot,\cdot\})$ defined with the same basis by
 \begin{eqnarray*}
\{e_1,e_2,e_3\}=e_2.
\end{eqnarray*}

Consider a linear map $R:\g\rightarrow\g$ defined by
$\left(\begin{array}{ccc}
a_{11}&a_{12}&a_{13}\\
a_{21}&a_{22}&a_{23}\\
a_{31}&a_{32}&a_{33}\end{array}\right)$ with respect to the basis $\{e_1,e_2,e_3\}$. Define
$$Re_1=a_{11}e_1+a_{21}e_2+a_{31}e_3,\quad Re_2=a_{12}e_1+a_{22}e_2+a_{32}e_3,\quad Re_3=a_{13}e_1+a_{23}e_2+a_{33}e_3.$$
Then $R$ is a Reynolds operator on the Lie algebra $\g$ if and only if
\begin{eqnarray*}
[Re_i,Re_j]=R[Re_i,e_j]+R[e_i,Re_j]-R[Re_i,Re_j],\quad i,j=1,2,3.
\end{eqnarray*}
By a straightforward computation, we conclude that $R$ is a Reynolds operator on the Lie algebra $\g$ if and only if
\begin{eqnarray*}
\begin{aligned}
(a_{21}a_{13}-a_{11}a_{23}+a_{23})a_{12}&=0;&(a_{22}a_{13}-a_{12}a_{23}-a_{13})a_{12}&=0;\\
(a_{21}a_{13}-a_{11}a_{23}+a_{23})a_{32}&=0;&(a_{22}a_{13}-a_{12}a_{23}-a_{13})a_{32}&=0;\\
(a_{11}+a_{22}-a_{11}a_{22}+a_{21}a_{12})a_{12}&=0;&(a_{11}+a_{22}-a_{11}a_{22}+a_{21}a_{12})a_{32}&=0;\\
(a_{11}a_{23}-a_{21}a_{13})(a_{22}+1)&=a_{23}a_{22}; &(a_{22}a_{13}-a_{12}a_{23})(a_{22}+1)&=a_{13}a_{22};\\
(a_{11}+a_{22}-a_{11}a_{22}+a_{21}a_{12})a_{22}&=a_{11}a_{22}-a_{21}a_{12}.
\end{aligned}
\end{eqnarray*}

By Theorem \ref{n+1Reynold}, if $R$ is also a Reynolds operator on the $3$-Lie algebra $\g_{f}$, then $R$ satisfies
\begin{eqnarray*}
f(e_1)R[Re_2,Re_3]+f(e_2)R[Re_3,Re_1]+f(e_3)R[Re_1,Re_2]=0,
\end{eqnarray*}
that is
\begin{eqnarray*}
(a_{12}a_{23}-a_{22}a_{13}+a_{11}a_{22}-a_{21}a_{12})a_{12}=0;\\
(a_{12}a_{23}-a_{22}a_{13}+a_{11}a_{22}-a_{21}a_{12})a_{22}=0;\\
(a_{12}a_{23}-a_{22}a_{13}+a_{11}a_{22}-a_{21}a_{12})a_{32}=0.
\end{eqnarray*}
So we have the following two cases to consider.
\begin{itemize}
\item[{\rm (ii)}]
If $a_{12}=a_{22}=a_{32}=0,$ then we deduce that $a_{13}=\frac{a_{11}a_{23}}{a_{21}}.$
\item[{\rm (ii)}]
If\begin{eqnarray*}
\left\{\begin{array}{rcl}
{}a_{21}a_{13}-a_{11}a_{23}+a_{23}=0;\\
{}a_{22}a_{13}-a_{12}a_{23}-a_{13}=0;\\
{}a_{11}+a_{22}-a_{11}a_{22}+a_{21}a_{12}=0;\\
{}a_{12}a_{23}-a_{22}a_{13}+a_{11}a_{22}-a_{21}a_{12}=0.
\end{array}\right.
\end{eqnarray*}
then we deduce that $a_{11}=-a_{22},a_{21}=-\frac{a^{2}_{22}}{a_{12}},a_{13}=a_{23}=0.$
\end{itemize}
Therefore, we can obtain
$$\left(\begin{array}{ccc}
a_{11}&0&\frac{a_{11}a_{23}}{a_{21}}\\
a_{21}&0&a_{23}\\
a_{31}&0&a_{33}
\end{array}\right)
 \quad\mbox{and}\quad
\left(\begin{array}{ccc}
-a_{22}&a_{12}&0\\
-\frac{a^2_{22}}{a_{12}}&a_{22}&0\\
a_{31}&a_{32}&a_{33}\end{array}\right)$$ are both  Reynolds operators on the Lie algebra $\g$ and the $3$-Lie algebra $\g_{f}$.}
%By Corollary \ref{3-Lie-Reynold-algebra}, the Reynolds Lie algebra induces a new Reynolds $3$-Lie algebra.
\end{ex}

 \subsection{Constructions of Reynolds operators on $3$-Lie algebras from Reynolds operators on commutative associative algebras}
In this subsection, we construct Reynolds $3$-Lie algebras from commutative associative Reynolds algebras.
Recall from \cite{Das-2} that a Reynolds operator on an associative algebra $(\g,\cdot)$ is a linear map $R:\g\rightarrow\g$ satisfying
\begin{equation}\label{eq:Reynolds LA}
Rx\cdot Ry=R(Rx\cdot y+x\cdot Ry-Rx\cdot Ry),\quad \forall~x,y\in\g.
\end{equation}
\begin{lem}\label{lem:comm asso alg}{\rm(\cite{BaiRGuo})}
Let $(\g,\cdot)$ be a commutative associative algebra. Let $D\in\Der(\g)$ and $f\in \g^*$ satisfy
$f(D(x)\cdot y)=f(x\cdot D(y)).$
Then $(\g,\{\cdot,\cdot,\cdot\}_{f,D})$ is a $3$-Lie algebra, where the bracket is given by
\begin{eqnarray}\label{eq:Rota-Baxter 1}
&&\{x,y,z\}_{f,D}\triangleq \begin{vmatrix}f(x)& f(y)&f(z)\\ D(x)&D(y)&D(z)\\ x&y&z\end{vmatrix}\\
&&\nonumber\triangleq f(x)(D(y)\cdot z-D(z)\cdot y)+f(y)(D(z)\cdot x-D(x)\cdot z)+f(z)(D(x)\cdot y-D(y)\cdot x).
\end{eqnarray}
\end{lem}

\begin{pro}
Let $(\g,\cdot,R)$ be a Reynolds commutative associative  algebra. Let $D\in \Der(\g)$ satisfying $DR=RD$ and $f\in \g^*$ satisfy $f(D(x)\cdot y)=f(x\cdot D(y)).$ Then $R$ is a Reynolds operator on the $3$-Lie algebra $(\g,\{\cdot,\cdot,\cdot\}_{f,D})$ if and only if $R$ satisfies
\begin{eqnarray}\label{eq:Reynolds-1-operator}
\begin{vmatrix}f(x)& f(y)&f(z)\\
 D(Rx)&D(Ry)&D(Rz)\\
  Rx&Ry&Rz
  \end{vmatrix}
=0.
\end{eqnarray}
\end{pro}
\begin{proof}
For all $x,y\in\g$, define
$[x,y]_D=D(x)\cdot y-D(y)\cdot x.$
By direct calculations,  we can verify that $(\g,[\cdot,\cdot]_D)$ is a Lie algebra. Furthermore, assume that $R$ is a Reynolds operator on $(\g,\cdot)$ satisfying $DR=RD$. Then  we have
\begin{eqnarray*}
[Rx,Ry]_D&=&D(Rx)\cdot Ry-Rx\cdot D(Ry)\\
&=&R(Dx)\cdot Ry-Rx\cdot R(Dy)\\
&=&R\Big(R(Dx)\cdot y+Dx\cdot Ry-R(Dx)\cdot Ry\Big)-R\Big(Rx\cdot Dy+x\cdot R(Dy)-Rx\cdot R(Dy)\Big)\\
&=&R\Big(D(Rx)\cdot y-Dy\cdot Rx\Big)+R\Big(Dx\cdot Ry-D(Ry)\cdot x\Big) -R\Big(D(Rx)\cdot Ry-D(Ry)\cdot Rx\Big)\\
&=&R([Rx,y]_D+[x,Ry]_D-R[Rx,Ry]_D),
\end{eqnarray*}
which implies that $R$ is a Reynolds operator on the Lie algebra $(\g,[\cdot,\cdot]_D)$.
By Theorem \ref{n+1Reynold}, $R$ is a Reynolds operator on the $3$-Lie algebra $(\g,\{\cdot,\cdot,\cdot\}_{f,D})$ if and only if \eqref{eq:Reynolds-1-operator} holds .
\end{proof}

Let $(\g,\cdot)$ be a commutative associative algebra.
For $x_i,~y_i,~z_i\in \g,~i=1,~2,~3$, denote by
\begin{eqnarray*}
\begin{vmatrix}\vec{x}&\vec{y}&\vec{z}\end{vmatrix}&=&\begin{vmatrix}x_1& y_1&z_1\\ x_2&y_2&z_2\\ x_3&y_3&z_3\end{vmatrix}\\
&=&x_1\cdot(y_2\cdot z_3-y_3\cdot z_2)-x_2\cdot(y_1\cdot z_3-y_3\cdot z_1)+x_3\cdot(y_1\cdot z_2-y_2\cdot z_1),
\end{eqnarray*}
where $\vec{x}$, $\vec{y}$ and $\vec{z}$ denote the column vectors.

\begin{lem}\label{lem:calcu det}
Let $R$ be a Reynolds operator on a commutative associative algebra $(\g,\cdot)$ and $R(\vec{x}),R(\vec{y}),R(\vec{z})$ denote the images of the column vectors. Then we have
\begin{eqnarray*}
\begin{vmatrix}R(\vec{x})&R(\vec{y})&R(\vec{z})\end{vmatrix}&=&R\left(\begin{vmatrix}R(\vec{x})&R(\vec{y})&\vec{z}\end{vmatrix}+c.p.\right)-R\left(\begin{vmatrix}R(\vec{x})&R(\vec{y})&R(\vec{z})\end{vmatrix}\right).
\end{eqnarray*}
\end{lem}
\begin{proof}
Since $R$ is a Reynolds operator on $(\g,\cdot)$, we have
 \begin{eqnarray*}
\begin{vmatrix}R(\vec{x})&R(\vec{y})&R(\vec{z})\end{vmatrix}&=&\sum_{\sigma\in S_3}\sgn(\sigma)R(x_{\sigma(1)})R(y_{\sigma(2)})R(z_{\sigma(3)})\\
&=&R\big(\sum_{\sigma\in S_3}\sgn(\sigma)R(x_{\sigma(1)})R(y_{\sigma(2)})z_{\sigma(3)}+c.p.\big)\\&&
-R\big(\sum_{\sigma\in S_3}\sgn(\sigma)R(x_{\sigma(1)})R(y_{\sigma(2)})R(z_{\sigma(3)})\big)\\
&=&R\left(\begin{vmatrix}R(\vec{x})&R(\vec{y})&\vec{z}\end{vmatrix}+c.p.\right)-R\left(\begin{vmatrix}R(\vec{x})&R(\vec{y})&R(\vec{z})\end{vmatrix}\right).
\end{eqnarray*}
The proof is finished.
\end{proof}
\begin{lem}\label{lem:2derivation}{\rm(\cite{Dzhu-1})}
Let $(\g,\cdot)$ be a commutative associative algebra, $D_1,D_2\in\Der(\g)$ satisfy $D_1D_2=D_2D_1$. Then
$(\g,\llbracket \cdot,\cdot,\cdot\rrbracket)$ is a $3$-Lie algebra, where the bracket is given by
\begin{eqnarray}\label{eq:Rota-Baxter 2}
\llbracket x_1,x_2,x_3\rrbracket\triangleq \begin{vmatrix}x_1& x_2&x_3\\ D_1(x_1)&D_1(x_2)&D_1(x_3)\\ D_2(x_1)&D_2(x_2)&D_2(x_3)\end{vmatrix},\quad\forall x_1,x_2,x_3\in \g.
\end{eqnarray}
\end{lem}

\begin{pro}\label{pro:ex2}
Let $(\g,\cdot,R)$ be a Reynolds commutative associative algebra. Let $D_1,D_2\in\Der(\g)$  satisfy $D_1D_2=D_2D_1$ and $RD_1=D_1R,~ RD_2=D_2R.$
Then $R$ is a Reynolds operator on the $3$-Lie algebra $(\g,\llbracket \cdot,\cdot,\cdot\rrbracket)$, where the bracket $\llbracket \cdot,\cdot,\cdot\rrbracket:\wedge^3\g\rightarrow\g$ is given by  $\eqref{eq:Rota-Baxter 2}$.
\end{pro}

\begin{proof} Since $RD_1=D_1R,\  RD_2=D_2R$, by Lemma \ref{lem:calcu det}, we have
\begin{eqnarray*}
\llbracket Rx_1,Rx_2,Rx_3\rrbracket&=& \begin{vmatrix}Rx_1& Rx_2&Rx_3\\ D_1(Rx_1)&D_1(Rx_2)&D_1(Rx_3)\\ D_2(Rx_1)&D_2(Rx_2)&D_2(Rx_3)\end{vmatrix}\\
&=&\begin{vmatrix}Rx_1& Rx_2&Rx_3\\ RD_1(x_1)&RD_1(x_2)&RD_1(x_3)\\ RD_2(x_1)&RD_2(x_2)&RD_2(x_3)\end{vmatrix}\\
&=&R\left(\begin{vmatrix}Rx_1& Rx_2&x_3\\ RD_1(x_1)&RD_1(x_2)&D_1(x_3)\\ RD_2(x_1)&RD_2(x_2)&D_2(x_3)\end{vmatrix}+c.p.\right)\\
&&-R\left(\begin{vmatrix}Rx_1& Rx_2&Rx_3\\ RD_1(x_1)&RD_1(x_2)&RD_1(x_3)\\ RD_2(x_1)&RD_2(x_2)&RD_2(x_3)\end{vmatrix}\right)\\
&=&R\left(\begin{vmatrix}Rx_1& Rx_2&x_3\\ D_1(Rx_1)&D_1(Rx_2)&D_1(x_3)\\ D_2(Rx_1)&D_2(Rx_2)&D_2(x_3)\end{vmatrix}+c.p.\right)\\
&&-R\left(\begin{vmatrix}Rx_1& Rx_2&Rx_3\\ D_1(Rx_1)&D_1(Rx_2)&D_1(Rx_3)\\ D_2(Rx_1)&D_2(Rx_2)&D_2(Rx_3)\end{vmatrix}\right)\\
&=&R(\llbracket Rx_1,Rx_2,x_3\rrbracket+c.p.)-R(\llbracket Rx_1,Rx_2,Rx_3\rrbracket).
\end{eqnarray*}
Thus, $R$ is a Reynolds operator on the $3$-Lie algebra  $(\g,\llbracket \cdot,\cdot,\cdot\rrbracket)$.
\end{proof}

\begin{lem}\label{lem:3derivation}{\rm(\cite{Filippov})}
Let $(\g,\cdot)$ be a commutative associative algebra. Let $D_i\in\Der(\g)$ such that $D_iD_j=D_jD_i,\ i,j=1,2,3$. Then
$(\g,\llbracket \cdot,\cdot,\cdot\rrbracket)$ is a $3$-Lie algebra, where the bracket is given by
\begin{eqnarray}\label{eq:Rota-Baxter 3}
\llbracket x,y,z\rrbracket:= \begin{vmatrix}D_1(x)&D_1(y)&D_1(z)\\ D_2(x)&D_2(y)&D_2(z)\\ D_3(x)&D_3(y)&D_3(z)\end{vmatrix},\quad \forall x,y,z\in\g.
\end{eqnarray}
\end{lem}

\begin{pro}
Let $(\g,\cdot,R)$ be a Reynolds commutative associative algebra, $D_1,D_2,D_3$ be derivations of $(\g,\cdot)$ satisfying $D_iD_j=D_jD_i$ and $RD_i=D_iR$ for $\ i,j=1,2,3, i\neq j$. Then $R$ is a Reynolds operator on the $3$-Lie algebra $(\g,\llbracket \cdot,\cdot,\cdot\rrbracket)$, where the bracket $\llbracket \cdot,\cdot,\cdot\rrbracket:\wedge^3\g\rightarrow\g$ is given by $\eqref{eq:Rota-Baxter 3}$.
\end{pro}

\begin{proof}
The proof is similar to the proof of Proposition \ref{pro:ex2}.
\end{proof}

 \end{document}